\pdfoutput=1
\documentclass[a4paper,12pt,reqno]{article}
\usepackage[markup=underlined]{changes}
\RequirePackage{etex}

\usepackage[T1]{fontenc}
\usepackage{fourier}
\usepackage{newcent}
\usepackage{microtype}
\DeclareMathSizes{12}{13}{8}{6}

\usepackage{etex}

\usepackage{geometry}
\usepackage{setspace}
\usepackage{titlesec}
\usepackage{amsmath,amssymb,amsthm,mathtools}
\usepackage{array,booktabs}
\usepackage{enumitem}
\usepackage[ruled,vlined,linesnumbered,algo2e]{algorithm2e}
\usepackage[table]{xcolor}
\usepackage{colortbl}

\definecolor{payoffgray}{gray}{0.92}

\usepackage{multirow}

\usepackage{bm}
\usepackage{makecell}
\usepackage{tabularx}
\usepackage{MnSymbol}
\expandafter\let\expandafter\IVFEWidehatInternal\csname widehat \endcsname
\expandafter\let\expandafter\IVFEWidetildeInternal\csname widetilde \endcsname
\AtBeginDocument{%
  \expandafter\let\csname widehat \endcsname\IVFEWidehatInternal
  \expandafter\let\csname widetilde \endcsname\IVFEWidetildeInternal
}

\usepackage{natbib}
\usepackage{tikz}
\usepackage[framemethod=TikZ]{mdframed}
\usepackage{changepage}
\usepackage{thm-restate}

\makeatletter
\@ifundefined{newcounteralias}{}{%
  \renewcommand\thmt@autorefsetup{%
    \@xa\def\csname\thmt@envname autorefname\@xa\endcsname\@xa{\thmt@thmname}%
  }%
}
\makeatother

\DeclarePairedDelimiterX{\inp}[2]{\langle}{\rangle}{#1, #2}

\usepackage[hypertexnames=false]{hyperref}
\usepackage[nameinlink,capitalise,noabbrev]{cleveref}

\definecolor{niceMagenta}{HTML}{C2185B}
\definecolor{softIndigo}{HTML}{3949AB}
\definecolor{oceanTeal}{HTML}{00897B}
\definecolor{warmOrange}{HTML}{BF4D00}
\definecolor{slateGray}{HTML}{546E7A}
\definecolor{purple}{HTML}{7B1FA2}
\definecolor{acmPurple}{HTML}{65016B}

\hypersetup{
  colorlinks=true,
  linkcolor=acmPurple,
  citecolor=acmPurple,
  urlcolor=black,
  pdftitle={The Complexity of Information-Value-Free Equilibria and Regulating Algorithmic Collusion},
  pdfsubject={Computational and query complexity of information-value-free equilibria}
}

\allowdisplaybreaks

\theoremstyle{plain}
\newtheorem{theorem}{Theorem}[section]
\newtheorem{lemma}{Lemma}[section]

\newtheorem{proposition}{Proposition}[section]
\newtheorem{corollary}{Corollary}[section]
\newtheorem{informaltheorem}{Main Result}
\newtheorem{informalcorollary}{Corollary}
\crefname{informaltheorem}{Main Result}{Main Results}
\Crefname{informaltheorem}{Main Result}{Main Results}
\crefname{informalcorollary}{Corollary}{Corollaries}
\Crefname{informalcorollary}{Corollary}{Corollaries}
\AtBeginEnvironment{informaltheorem}{\setlength{\emergencystretch}{1.5em}}
\AtBeginEnvironment{informalcorollary}{\setlength{\emergencystretch}{1.5em}}

\theoremstyle{definition}
\newtheorem{definition}{Definition}[section]

\newtheorem{example}{Example}[section]

\theoremstyle{remark}

\definecolor{mainResultBorder}{HTML}{356B82}
\definecolor{mainResultBackground}{HTML}{EEF5F7}
\definecolor{challengeBorder}{HTML}{A45F2A}
\definecolor{challengeBackground}{HTML}{FBF2E8}

\mdfdefinestyle{mainResultStyle}{%
  backgroundcolor=mainResultBackground,
  linecolor=mainResultBorder,
  linewidth=0pt,
  roundcorner=0pt,
  skipabove=8pt,
  skipbelow=8pt,
  innertopmargin=0pt,
  innerbottommargin=7pt,
  innerleftmargin=9pt,
  innerrightmargin=9pt
}

\mdfdefinestyle{challengeStyle}{%
  backgroundcolor=challengeBackground,
  linecolor=challengeBorder,
  linewidth=0pt,
  roundcorner=0pt,
  skipabove=8pt,
  skipbelow=8pt,
  innertopmargin=6pt,
  innerbottommargin=6pt,
  innerleftmargin=9pt,
  innerrightmargin=9pt,
  nobreak=true,
  footnoteinside=false
}

\newcounter{challenge}

\usepackage{mathtools}
\titlespacing*{\paragraph}{0pt}{1.25ex plus 1ex minus .2ex}{0.5em}

\titleformat{\section}
		{\bfseries\center}
         {\thesection}
        {0.5em}
        {}
        []

\titleformat{\subsection}[runin]
        {\normalfont\bfseries}
        {\thesubsection}
        {0.5em}
        {\addperiod}
        []
\newcommand{\addperiod}[1]{#1.}

\newcommand{\E}{\mathbb{E}}
\newcommand{\Prb}{\mathbb{P}}
\newcommand{\Q}{\mathbb{Q}}
\newcommand{\R}{\mathbb{R}}
\newcommand{\Z}{\mathbb{Z}}
\newcommand{\payoffdenom}{B}
\newcommand{\one}{\boldsymbol{1}}
\newcommand{\ind}[1]{\one\!\left\{#1\right\}}
\newcommand{\pos}[1]{\left[#1\right]_{+}}
\newcommand{\supp}{\operatorname{supp}}
\newcommand{\dist}{\operatorname{dist}}
\newcommand{\poly}{\operatorname{poly}}

\newcommand{\defeq}{\coloneqq}
\newcommand{\argmaxop}{\operatorname*{arg\,max}}
\newcommand{\Reg}{\mathrm{Reg}}
\newcommand{\SwapReg}{\mathrm{SwapReg}}
\newcommand{\barreg}{\overline{\Reg}}
\newcommand{\barswapreg}{\overline{\SwapReg}}
\newcommand{\IVFCE}{\mathrm{IVFCE}}
\newcommand{\IVFCCE}{\mathrm{IVFCCE}}

\newcommand{\PPAD}{\mathsf{PPAD}}
\newcommand{\PLS}{\mathsf{PLS}}
\newcommand{\CLS}{\mathsf{CLS}}

\let\oldcitet\citet
\let\oldcitep\citep
\renewcommand{\citet}[1]{\oldcitet*{#1}}
\renewcommand{\citep}[1]{\oldcitep*{#1}}

\title{\Large\scshape\bfseries
Algorithmic Collusion and the Complexity\\
of Information-Value-Free Equilibria}
\date{}

\author{
\makebox[.27\linewidth]{Ioannis Anagnostides\thanks{Carnegie Mellon University; email: \protect\texttt{ianagnos@cmu.edu}}}\\
CMU
\and
\makebox[.27\linewidth]{Weiqiang Zheng\thanks{Yale University; email: \protect\texttt{weiqiang.zheng@yale.edu}}}\\
Yale
}

\begin{document}
\pagenumbering{roman}
\pagestyle{empty}
\maketitle

\begin{abstract}
A (coarse) correlated equilibrium (CE) is \emph{information-value-free (IVF)} if a player can match the payoff obtained from recommendations by committing to a fixed action. Motivated by the problem of regulating algorithmic collusion, this refinement was introduced by Hartline, Wang, and Zhang [EC'26], who showed that it can be computed in polynomial time in explicitly represented normal-form game. In this paper, we examine the complexity of IVF(C)CEs in succinct games, which model more realistic strategic interactions that feature either many players or exponentially many pure strategies.

We first show that computing an information-value-free CE is PPAD-complete---that is, polynomial-time equivalent to Nash equilibrium---in many-player polymatrix games or two-player Bayesian games, even when the approximation is a constant. We also prove an unconditional exponential query lower bound. Our results establish that IVFCEs are intractable, \emph{even in the centralized model}, and rule out the existence of any efficient learning dynamics. This significantly strengthens the impossibility result of Hartline, Wang, and Zhang, which concerns a particular class of learning algorithms, and furnishes strong computational critiques of recent regulation on algorithmic collusion.

To sidestep these hardness results, we examine the complexity of information-value-free CCE. Certain no-regret algorithms---such as regret matching or FTRL---provide a fully polynomial-time approximation scheme (FPTAS) for this problem. The complexity when the approximation is exponentially small turns out to be nuanced. On the one hand, leveraging no-regret dynamics, we establish membership in \(\text{CLS} = \text{PPAD} \cap \text{PLS} \). On the other hand, we show that it is at least as hard as the P-matrix linear complementarity problem, and hence as hard as simple stochastic games. This shows that even IVFCCEs are unlikely to admit a polynomial-time algorithm barring a major breakthrough.
\end{abstract}
\thispagestyle{empty}

\clearpage

\setcounter{tocdepth}{2}
\tableofcontents
\thispagestyle{empty}

\clearpage

\pagenumbering{arabic}
\pagestyle{plain}

\section{Introduction}

Firms increasingly delegate pricing decisions to sophisticated algorithms that adjust to market conditions using data on sales and competitors' prices. Algorithmic pricing promises to enable faster and more effective responses to fluctuations in supply and demand. However, there is growing evidence that it can have the concomitant effect of increasing firms' margins, potentially harming consumers~\citep{Assad24:Algorithmic}. These findings reinforce concerns about \emph{algorithmic collusion}, whereby interacting pricing algorithms can sustain supracompetitive prices without explicit communication or an overt agreement to collude~\citep{Calvano20:Artificial}. This type of autonomous collusion falls outside the purview of traditional antitrust prohibitions on anticompetitive practices, making the development of effective regulatory approaches an important challenge~\citep{Harrington18:Developing,Chassang23:Regulating}.

A proposed remedy is to require from pricing algorithms to be \emph{rationalizable}, in the sense that they incur vanishing \emph{swap regret}, a condition that regulators can verify \emph{ex post}~\citep{Hartline24:Regulation,Hartline25:Regulation}. Swap regret is a classic notion in online learning ~\citep{Stoltz05:Internal,Blum07:External}, inextricably linked to \emph{correlated equilibrium}, a seminal game-theoretic solution concept~\citep{Aumann74:Subjectivity}. In particular, when every player employs a no-swap-regret algorithm, the empirical distribution of play approaches the set of correlated equilibria~\citep{Hart00:Simple,Foster97:Calibrated}. A correlated equilibrium can be modeled through a mediator, who privately recommends actions to the players so that no player can profitably deviate. A recommendation thus often provides a \emph{valuable signal}, revealing information about what the other players are likely to do.

The role of information in algorithmic pricing has recently attracted regulatory scrutiny. For example, the~\citet{DOJ25:RealPage} prescribes that ``firms should
not make pricing decisions using insight drawn from their
competitors’ nonpublic, competitively sensitive data.'' This raises the broader question of whether learning algorithms can reach correlated equilibria without \emph{generating or exchanging valuable information}. To formalize this, \citet{Hartline26:Computable} recently introduced the notion of \emph{information-value-freeness}. A correlated equilibrium is information-value-free when every player can match the equilibrium payoff by committing to a fixed action regardless of the recommendation. Thus, although the recommendation may contain information about other players' actions, it has no value to the player. The main focus of our paper is on information-value-free equilibria.

\paragraph{Model.} More concretely, we consider a finite normal-form game with a set of players \([n]\), finite action sets \(\mathcal A_i\), and utilities \(u_i:\mathcal A\to[0,1]\), where \(\mathcal A=\prod_{i=1}^n \mathcal A_i\) is the joint action set. For a (potentially correlated) distribution \(\mu\in\Delta(\mathcal A)\) over action profiles, player \(i\)'s average \emph{swap regret} is
\begin{align*}
 \barswapreg_i(\mu)
 \defeq
 \max_{\phi_i:\mathcal A_i\to\mathcal A_i}
 \E_{\boldsymbol a\sim\mu}
 \bigl[u_i(\phi_i(a_i),\boldsymbol a_{-i})-u_i(\boldsymbol a)\bigr].
\end{align*}
A distribution \(\mu\) is an \emph{\(\epsilon\)-correlated equilibrium (CE)} if \(\barswapreg_i(\mu)\le\epsilon\) for every player $i \in [n]$. Following~\citet{Hartline26:Computable}, we define the \emph{information-value} of player \(i\) under \(\mu\) by
\begin{equation}
 \operatorname{IV}_i(\mu)
 \defeq
 \pos{-\barreg_i(\mu)}
 =
 \pos{\E_{\boldsymbol{ a} \sim \mu}[u_i(\boldsymbol a)]
       -\max_{a_i'\in\mathcal A_i}
        \E_\mu[u_i(a_i',\boldsymbol a_{-i})]},
 \label{eq:iv-violation}
\end{equation}
where $\barreg_i(\mu)$ is the average \emph{external regret}. Equivalently, \(\operatorname{IV}_i(\mu)\le\epsilon\) iff \(\barreg_i(\mu)\ge-\epsilon\). We call \(\mu\) \emph{\(\epsilon\)-information-value-free} if \(\barreg_i(\mu)\ge-\epsilon\) for every player $i \in [n]$.

\begin{definition}[Information-value-free correlated equilibrium; \citealp{Hartline26:Computable}]
\label{def:ivf-equilibria}
A distribution \(\mu\) is an \emph{\((\epsilon,\delta)\)-IVFCE} if it is both an \(\epsilon\)-CE and \(\delta\)-information-value-free. We abbreviate \((\epsilon,\epsilon)\)-\(\IVFCE\) as \(\epsilon\)-\(\IVFCE\).
\end{definition}

An exact information-value-free CE always exists in finite games: every mixed Nash equilibrium~\citep{Nash51:Noncooperative} induces a product CE, which is indeed information-value-free. We provide some illustrative examples later in~\Cref{sec:preliminaries}.

\citet{Hartline26:Computable} observed that information-value-free correlated equilibrium can be computed in polynomial time in explicitly represented normal-form games.\footnote{\citet{Hartline26:Computable} pointed this out for two-player normal-form games, but it readily generalizes to all explicitly represented normal-form games with complexity \(O(m^n)\), where \(m\) is the maximum number of actions per player; this is polynomial in the explicit representation, which requires specifying \(m^n\) payoff entries for each player.} This can be done by enumerating the players' candidate best fixed actions and solving the corresponding CE linear programs. In stark contrast, they showed that it cannot be learned by a broad class of learning algorithms. More specifically, a no-swap-regret algorithm cannot in general make the empirical play converge to an information-value-free CE against a \emph{smooth learner}, who responds gradually to a sufficiently small payoff gap. Standard online algorithms such as multiplicative weights update, follow-the-perturbed-leader, online gradient ascent, and their Blum-Mansour reductions~\citep{Blum07:External} satisfy this property~\citep{Hartline26:Computable}. An implication of this impossibility result is that regulations that prohibit algorithms from creating or exchanging valuable information may also proscribe rationalizable learning, rather than restricting collusion. However, \citet{Hartline26:Computable} left open the possibility that non-smooth learners may converge to information-value-free CE.
\begin{quote}
    \centering
    \emph{Are there efficient learning dynamics that converge in polynomially many rounds to an information-value-free CE?}
\end{quote}

To shed light on this question, we examine the complexity of information-value-free CE beyond explicitly represented normal-form games, in \emph{succinct games}. Indeed, the LP algorithm noted by~\citet{Hartline26:Computable} scales \emph{exponentially} with the number of players. Yet, in succinct games, a correlated equilibrium can be computed in polynomial time through the famous ellipsoid against hope algorithm~\citep{Papadimitriou08:Computing}, and efficiently learned through online learning dynamics~\citep{Foster97:Calibrated,Hart00:Simple,CesaBianchi06:Prediction}. Succinct games model more realistic strategic interactions that feature either many players or exponentially many pure strategies. A canonical example is \emph{Bayesian games}, in which a player has to condition its action based on some observed signal. This raises the question: \emph{what is the complexity of information-value-free CE?} In particular, we are interested in how the additional requirement of information-value-freeness affects the complexity of CE. Besides the implications for algorithmic collusion, information-value-freeness is a natural concept that has been used in other settings, such as mechanism design (\Cref{sec:related}), and understanding its complexity is an important question.

\subsection{Our results}

We characterize the complexity of information-value-free equilibria. Taken together, our results provide further critique for regulating against learning algorithms that generate valuable information.

\subsubsection{Information-value-free correlated equilibrium}

We first show that, surprisingly, computing an information-value-free correlated equilibrium in succinct games is as hard as computing a Nash equilibrium~\citep{Daskalakis09:Complexity,chen09:Settling,Rubinstein15:Inapproximability}.

\begin{restatable}{informaltheorem}{mainppad}
\label{informal:ivfce-complexity}
There are universal constants \(\epsilon > 0 \) and \(m\) such that computing an \(\epsilon\)-IVFCE is \(\PPAD\)-hard even in bipartite polymatrix games of maximum degree three with at most \(m\) actions per player.
\end{restatable}

Assuming \(\PPAD\ne\mathsf{P}\), this rules out the existence of \emph{any} computationally efficient learning dynamics approaching the set of IVFCEs. It also precludes the existence of a \emph{polynomial-time approximation scheme (PTAS)}.\footnote{For a fixed $\epsilon > 0$, a PTAS computes an \(\epsilon\)-approximate solution in time polynomial in the problem's input size.}

The same limitation applies even in \emph{two-player Bayesian games}, where each player can condition actions based on an observed type.\footnote{An agent-normal-form correlated equilibrium is a correlated equilibrium of the game obtained by treating each player-type as a separate player, who observes only the recommendation associated with that type~\citep{Forges93:Five}.}

\begin{informalcorollary}
\label{informal:ivfce-bayesian}
Computing an \(\epsilon\)-information-value-free CE of the agent normal form of a two-player Bayesian game is \(\PPAD\)-hard, even when the game has polynomially many types and at most \(m\) actions per type.
\end{informalcorollary}

This follows from~\Cref{informal:ivfce-complexity} by interpreting the vertices on each side of the bipartite polymatrix graph as the types of one Bayesian player, and drawing an edge uniformly as the type profile.

Moreover, we establish an \emph{unconditional} exponential query lower bound for computing an $\epsilon$-IVFCE even when $\epsilon$ is a constant.

\begin{restatable}{informaltheorem}{mainquery}
\label{informal:ivfce-query}
For universal constants \(\epsilon>0\) and \(m\), every randomized algorithm that computes an \(\epsilon\)-information-value-free CE of an \(n\)-player, \(m\)-action game with probability at least \(2/3\) requires \(2^{\Omega(n)}\) payoff queries in the worst case.
\end{restatable}

\Cref{informal:ivfce-query} complements the impossibility result by~\citet{Hartline26:Computable}, who showed that \emph{smooth learners} cannot converge to IVFCE. Our result shows that \emph{any} learning algorithms---with potentially unbounded computation and centralized coordination---require exponentially many iterations to converge to a constant-precision IVFCE.



\subsubsection{Information-value-free coarse correlated equilibrium}

These hardness results show that IVFCEs cannot be attained in the worst case. This raises the question of whether there are meaningful relaxations of CEs that are compatible with information-value-freeness without sacrificing polynomial-time computation. To address this, we consider the relaxed notion of \emph{coarse} correlated equilibrium (CCE)~\citep{Moulin78:Strategically}. We recall that a distribution \(\mu\) is an \emph{\(\epsilon\)-CCE} if \(\barreg_i(\mu)\le\epsilon\) for every player. Ae before, we say that \(\mu\) is \(\epsilon\)-information-value-free if \(\barreg_i(\mu)\ge-\epsilon\) for every player, following the same convention as in~\Cref{def:ivf-equilibria}. That is, an IVFCCE prescribes $\barreg_i(\mu) \approx 0$.

In contrast to IVFCEs, there are no-external-regret algorithms that guarantee nonnegative regret. In particular, such dynamics yield a fully polynomial-time approximation scheme (FPTAS) for \((\epsilon, 0) \)-IVFCCE (\Cref{thm:ivfcce-fptas}). 

This in fact follows from existing results in the online learning literature (\emph{e.g.},~\citealp{Guzman21:Best}), who examined this question with a different motivation in mind. This is in contrast to IVFCE, for which there is no PTAS unless $\PPAD  = \mathsf{P}$ (\Cref{informal:ivfce-complexity}) and it is information-theoretically impossible to simultaneously guarantee sublinear swap regret and nonnegative external regret~\citep{Hartline26:Computable}. More specifically, we give two examples. The first is \emph{regret matching}~\citep{Hart00:Simple}, which maintains nonnegative external regret (under a specific tie-breaking rule). The second is the class of \emph{follow-the-regularized-leader (FTRL)} algorithms, which has been shown to have nonnegative (or sublinearly negative when the regularizer is adaptive) external regret~\citep{Guzman21:Best}. 

On top of this, FTRL algorithms are also \emph{mean-based}~\citep{braverman2018selling}, which implies that they iteratively eliminate dominated actions. In fact, a recent work by~\citet{Nisan26:Domination} argues that \emph{dominance-avoiding} algorithms like FTRL are---in some sense---non-collusive, converging to competitive prices in Bertrand duopoly games. This viewpoint is supported by~\citet{Durmann26:Online}. Therefore, FTRL is a class of algorithms that can guarantee simultaneously no external regret, information-value-freeness, and dominance-avoidance.

\definecolor{tableblue}{HTML}{1F4E79}
\definecolor{headergray}{HTML}{F1F2F4}
\definecolor{ivfblue}{HTML}{EAF3F8}

\renewcommand{\tabularxcolumn}[1]{m{#1}}

\begin{table}[t]
    \centering
    \caption{Computational complexity of equilibrium notions in
    normal-form games: Nash equilibrium (NE), correlated equilibrium (CE), coarse correlated equilibrium (CCE), information-value-free CE (IVFCE), and information-value-free CCE (IVFCCE). The results of this paper are highlighted in the shaded cells.}
    \label{tab:equilibrium-complexity}

    \renewcommand{\arraystretch}{1.15}
    \setlength{\tabcolsep}{6pt}

    \begin{tabularx}{\textwidth}{
        @{}
        >{\raggedright\arraybackslash}m{0.16\textwidth}
        >{\centering\arraybackslash}m{0.35\textwidth}
        >{\centering\arraybackslash}X
        @{}
    }
        \arrayrulecolor{black}
        \toprule

        \textcolor{black}{Equilibrium}
        &
        \textcolor{black}{Explicit normal form}
        &
        \textcolor{black}{Succinct normal form}
        \\

        \midrule

        NE
        &
        \shortstack{
            $\PPAD$-complete \\
            for $1/\poly$ approximation\\
            \citep{chen09:Settling}
        }
        &
        \shortstack{
            $\PPAD$-complete\\
            for constant approximation\\
            \citep{Rubinstein15:Inapproximability}
        }
        \\

        \arrayrulecolor{black}
        \midrule

       CE
        &
        &
        \\

        CCE
        &
        \multirow{-2}{*}{\shortstack{$\mathsf{P}$ \\ \citep{Aumann87:Correlated}} }
        &
        \multirow{-2}{*}{\shortstack{$\mathsf{P}$\\ \small \citep{Papadimitriou08:Computing} }}
        \\

        \midrule

        IVFCE
        &  \shortstack{$\mathsf{P}$\\ \footnotesize \citep{Hartline26:Computable} }
        &
        \cellcolor{ivfblue}
        \shortstack{
            $\PPAD$-complete\\
            for constant approximation
        }
        \\
       
        \midrule

        IVFCCE
        & \shortstack{$\mathsf{P}$\\ \footnotesize \citep{Hartline26:Computable} }
        &
        \cellcolor{ivfblue}
        \makecell[c]{FPTAS\\
            $\CLS$ membership\\
            $P$-LCP-hard}
        \\
        \arrayrulecolor{black}
        \bottomrule
    \end{tabularx}
\end{table}


\paragraph{Complexity beyond FPTAS.}
We then consider the complexity of $\epsilon$-IVFCCE in the high-accuracy regime, where the precision $\epsilon$ is exponentially small and we seek an algorithm that runs in $\poly(\log(1/\epsilon))$ time. This regime turns out to be surprisingly nuanced. On one hand, we show that the problem lies in $\CLS = \PPAD \cap \PLS$~\citep{Daskalakis11:Continuous,Fearnley23:Complexity}, a subclass of $\PPAD$ that captures the complexity of KKT points. As a consequence, this separates (barring complexity collapses) IVFCEs from IVFCCEs, even in the high-accuracy regime. On the other hand, we show that computing an IVFCCE is at least as hard as the  \emph{\(P\)-matrix linear complementarity problem}~\citep{Megiddo88:Note,Cottle09:Linear}. We summarize these key results concerning IFVCCEs below.

\begin{informaltheorem}
\label{informal:ivfcce-binary}
Computing an $\epsilon$-IVFCCE belongs to $\CLS$, but is at least as hard as the P-matrix  LCP problem.
\end{informaltheorem}

In particular, the \(P\)-matrix LCP problem subsumes \emph{simple stochastic games}~\citep{Gartner05:Simple}, and has eluded polynomial-time algorithms for decades~\citep{Fearnley20:Unique}. This provides evidence that computing even IVFCCEs in succinct games is unlikely to admit a polynomial-time algorithm. Characterizing the precise complexity of IVFCCEs remains a challenging open problem for future research; we conjecture that it is $\CLS$-complete.\footnote{GPT 6.0 Astra at maximum effort level was unable to make progress on this problem at the time of the writing of this paper.}


\subsection{Technical overview}

We now highlight the key technical insights behind our results, starting from IVFCEs. In our reductions, we embed a quadratic score in players' utilities. At a high-level, information-value-freeness means that the recommendations contain no information relative to the prior. At the same time, the recommendation is calibrated and close to the posterior because of the CE constraint. These two properties together imply that the posterior is close to the prior, which in turn implies that the recommendation itself has small variance. This creates an \emph{equilibrium collapse}: the marginals of any approximate IVFCE yields an approximate Nash equilibrium.

We now explain the argument in more detail. The construction is based on the \emph{imitation game}, which is commonly used in the complexity of equilibria~\citep{Mclennan10:Imitation,Mclennan10:Simple,Babichenko16:Query,Roughgarden16:Communication,Babichenko17:Communication,Etessami20:Tarski,Babichenko21:Settling,Goos21:Near}. In particular, the (normal-form) correlated equilibrium of this imitation game were recently analyzed by~\citet{Anagnostides26:Complexity}, who established contraction-hardness. Instead, here we show $\PPAD$-hardness even for a constant precision under the additional constraint of information-value-freeness.

In what follows, for the sake of exposition, we will make certain innocuous simplifications: i) the analysis will be carried out in one dimension, and ii) players will have continuous action spaces. The formal argument without these restrictions appears in~\Cref{sec:ivfce}.

In this context, let $u(x, y) = 1 - (x - y)^2$ be the utility function of the imitating player, where $x, y \in [0, 1]$. The goal of this player is to pick $x \approx y$. The other player has utility $1 - (y - f(x))^2$ and tries to pick $y \approx f(x)$ where $f: [0,1] \rightarrow [0,1]$ is a continuous function. Thus, a Nash equilibrium $(x^*, y^*)$ would give a fixed point $x^* = y^* = f(x^*)$; for high-dimensional $f$, this is known to be $\PPAD$-complete.

Now, let $\mu$ be a correlated distribution. By definition, we have
\begin{equation}
    \label{eq:CE-constraint}
    \E_{(x, y) \sim \mu} [(x-y)^2] - \E_{(x, y) \sim \mu} \left[ \min_{x' \in [0, 1]} (x' - y)^2 \mid x \right] \leq \barswapreg(\mu).
\end{equation}
In words, upon observing the recommendation $x$, that player will play the point $x'$ that minimizes the distance from $y$ in expectation. In particular, if $\mu(x)$ is the expectation of $y$ conditional on $x$, we have $\E_{(x, y) \sim \mu} [ \min_{x' \in [0, 1]} (x' - y)^2 \mid x ] = \E_{(x, y) \sim \mu} ( \mu(x) - y )^2 $. From this it follows that
\begin{equation}
    \tag{Calibration}
    \label{eq:calibration}
    \E_{(x, y) \sim \mu} [ (x - \mu(x) )^2 ] \leq \barswapreg(\mu),
\end{equation}
because
\[
    \E_{(x, y) \sim \mu} [ (x - \mu(x)) (y - \mu(x) ) ] = \E_{(x,y)} [ x - \mu(x)] \E_{(x, y) \sim \mu} [ y - \mu(x) \mid x ] = 0.
\]
Next, information-value-freeness implies
\[
    \operatorname{IV}(\mu) \geq \min_{x' \in [0, 1]} \E_{(x, y) \sim \mu} [ (x' - y)^2 ] - \E_{(x, y) \sim \mu} [(x-y)^2].
\]
Now, without observing the recommendation $x$, the best response is $\E[y]$; that is, $\min_{x' \in [0, 1]} \E_{(x, y) \sim \mu} [ (x' - y)^2 ] = \E_{(x,y) \sim \mu} [ (\E[y] - y)^2 ]$. Combining with~\eqref{eq:CE-constraint},
\[
    \E_{(x,y) \sim \mu} [ (\E[y] - y)^2 ] - \E_{(x,y) \sim \mu} [ (y - \mu(x))^2] \leq \operatorname{IV}(\mu) + \barswapreg(\mu),
\]
so
\begin{equation}
    \label{eq:IVFness}
    \E_{(x, y) \sim \mu} [ (\mu(x) - \E[y] )^2 ] \leq \operatorname{IV}(\mu) + \barswapreg(\mu).
\end{equation}
We will now argue that the variance of $x$, drawn from $\mu$, is small. Indeed, applying the AM-GM inequality,
\begin{align*}
    \E[ (x - \E[y])^2] &\leq 2 \E[ ( x - \mu(x))^2 ] + 2 \E[ ( \mu(x) - \E[y] )^2 ] \\
    &= 2 \E[ ( x - \mu(x))^2 ] + 2 \operatorname{Var}(\mu(x)),
\end{align*}
where the equality uses the fact that $\E[\mu(x)] = \E[y]$. Combining~\eqref{eq:calibration} and~\eqref{eq:IVFness}, it follows that the variance of $x$ is small and $\E[x] \approx \E[y]$, because $\E[ (x - \E[y])^2 ] = \operatorname{Var}(x) + ( \E[x] - \E[y] )^2 $. This implies that $\mu$ is essentially a product distribution, so one can readily extract a Nash equilibrium through its marginals.

Based on this analysis, \Cref{informal:ivfce-query} follows using existing exponential (black-box) lower bounds for Brouwer fixed points~\citep{Hirsch89:Exponential}, particularly the strengthening of~\citet{Rubinstein16:Settling}. For~\Cref{informal:ivfce-complexity}, the approach is similar, but the reduction starts from a polymatrix game~\citep{Rubinstein15:Inapproximability,Deligkas24:Pure}.

We next turn to the complexity of IVFCCEs when the precision is exponentially small.

\paragraph{Membership in polynomial local search.}

Our $\PLS$ membership proof relies on running suitable no-regret dynamics, which we previously saw yields an FPTAS. This is conceptually interesting, as we have not seen before applying no-regret dynamics to establish membership in $\PLS$. The first challenge is that the correlated equilibrium induced by no-regret dynamics has description growing linearly with the number of iterations, which is polynomial in $1/\epsilon$; when $\epsilon$ is exponentially small, this becomes prohibitive because the output of a $\PLS$ proof needs to have a polynomial description. We address this through a compression lemma (\Cref{lem:rational-compression}), which uses a constructive version of Carath\'eodory's theorem. Moreover, it is important to purify the mixed strategies prescribed by regret matching; otherwise, without the purification, the bit complexity would grow with each iteration. It is tempting here to simply sample from each mixed strategy. However, we need to make sure that the transition is deterministic in order to establish membership in $\PLS$. As formalized in~\Cref{lem:zero-drift-selection}, there is a simple way to deterministically pick pure strategies without affecting the regret upper bound.

Next, to establish the potential property, we attach a counter in the state space, indicating the time step of the dynamics. This step is delicate: simply adding a counter to an arbitrary search algorithm does not, of course, suffice, because checking the validity of a configuration could require backtracking over exponentially many steps in history. What makes our construction work is that the updates preserve directly checkable regret bounds (encoded explicitly in the state space), guaranteeing an IVFCCE at the final round for every valid state.

\paragraph{Hardness from \(P\)-matrix complementarity.}

We next turn to the hardness result. To put this context, we recall that there is a well-known equivalence between LCPs and Nash equilibria (\emph{e.g.},~\citealp{Mehta14:Constant}). Our reduction crucially hinges on having a \emph{$P$-matrix}, which means that every principal minor is positive. As usual in such reductions, each joint action profile in the constructed game induces a candidate \emph{basis} for the LCP. The key difference in our setting is that when considering a correlated distribution, we have a \emph{distribution} over candidate bases. The crucial role of the $P$-matrix property is that it allows averaging over bases \emph{via a consistent sign}; otherwise, these factors could cancel out, preventing recovery of a solution.

\subsection{Additional related literature}
\label{sec:related}

The correlated equilibrium was famously introduced by \citet{Aumann74:Subjectivity}. Our reduction makes use of a standard connection between CE and calibration, going back to~\citet{Foster97:Calibrated}, and further cultivated in recent years~\citep{Fishelson25:Full,Fishelson26:High,Casacuberta24:Complexity,Farina26:Efficient,Peng25:High}. From an algorithmic standpoint, a correlated equilibrium is a linear program, so it can be computed efficiently in polynomial time in explicitly represented games. In fact, \citet{Papadimitriou08:Computing} developed a framework for computing correlated equilibria even in succinct games of polynomial type, meaning that the number of actions available to each player is polynomial in the description. Subsequently, \citet{Jiang15:Polynomial} gave an exact algorithm for this problem. Our results concern the refinement resulting by information-value-freeness. Unlike ordinary correlated equilibria, which have a polynomial query complexity~\citep{Hart18:Query}, we show that IVFCEs share the same complexity fate as Nash equilibria.

\paragraph{Algorithmic collusion.} There has been a surge of research on algorithmic collusion, starting from~\citet{Calvano20:Artificial}. We mention here some (highly incomplete) highlights, and refer to citations therein for a more complete discussion. \citet{Banchio22:Artificial} develop a theoretical mechanism for explaining the emergence of algorithmic collusion, in which statistical linkage between the algorithms can produce coordinated behavior above static Nash equilibria. \citet{Arunachaleswaran25:Algorithmic} show that a seller committing to a no-regret algorithm can induce high prices when another seller optimizes against it, even without explicit punishment threats. Of particular relevance to us, \citet{Hansen21:Frontiers} show how \emph{correlated price} experimentation by independent algorithms can sustain high prices. Collusion can also arise in sequential pricing under Q-learning~\citep{Klein21:Autonomous} and has been observed with large language models~\citep{Fish24:Algorithmic}. We refer to the recent survey of~\citet{Hartline26:Economics} for further pointers and connections with no-regret learning.

\paragraph{Other uses of nonnegative regret.} \citet{Camara20:Mechanisms} (and subsequently~\citet{Collina24:Efficient}) formalize “uninformedness” by assuming that the agent’s external regret is non-negative. The fact that multiplicative weights guarantees nonnegative regret was observed, for example, by~\citet{Blum18:Preserving}, who leveraged that property in the context of multigroup fairness. A characterization of online algorithms with nonnegative regret was carried out earlier by~\citet{Gofer16:Lower}. This is contrast to certain incarnations of \emph{optimistic} mirror descent, which have the propensity to generate negative regret (\emph{e.g.},~\citealp{Hsieh21:Adaptive,Anagnostides22:Optimistic}).

\paragraph{More broadly about the value of information.} From a broader context, the role of information has been a mainstay in economics, and has more recently attracted much attention in the computer science literature in the context of information design. For some recent pointers, we refer to~\citet{Dughmi16:Algorithmic,Dughmi19:Persuasion,Celli20:Private,Castiglioni20:Online,Su2021:Bayesian,Yang24:Computational,Babichenko25:Information,Chen25:Explainable,Duetting25:Information,Ding23:Competitive,Berker26:Earn}.
\section{Additional preliminaries}
\label{sec:preliminaries}

In this section, we provide additional background and complexity conventions.

\subsection{Input and output representations}

In terms of representation, for the computational complexity results, we assume that the underlying game is given succinctly~\citep{Papadimitriou08:Computing}. In terms of the output, following~\citet{Papadimitriou08:Computing}, we assume that a correlated distribution should be given explicitly as a polynomial mixture of product distributions. It is worth pointing out that the IVFCE results do not require an explicit output representation, only a way to extract the marginals of the underlying distribution.

\paragraph{Polynomial expectation property.} Throughout this paper, we assume that the succinct game satisfies the \emph{polynomial expectation property}, which suffices for computing a CE in polynomial time~\citep{Papadimitriou08:Computing}. More specifically, let $L$ be the input length. The polynomial expectation property asserts that, given rational strategies \(\boldsymbol x_j\in\Delta(\mathcal A_j)\) and an action \(a_i'\in\mathcal A_i\), the expectation $\E_{\boldsymbol a_{-i}\sim\bigotimes_{j\ne i}\boldsymbol x_j}
 \bigl[u_i(a_i',\boldsymbol a_{-i})\bigr]$ can be computed exactly in polynomial time. In~\Cref{sec:CLS}, we denote by \(\payoffdenom \ge 1\) an integer encoded with \(\poly(L)\) bits such that \(u_i(\boldsymbol a')\in \payoffdenom^{-1} \Z \cap[0,1]\) for every player \(i\) and action profile \(\boldsymbol a'\).
 
\paragraph{Polymatrix games.} \emph{Polymatrix games}~\citep{Cai16:Zero}, and more broadly, bounded-degree graphical games, satisfy the polynomial expectation property. In particular, a polymatrix game is specified by an undirected graph \(([n],E)\).  We will use the shorthand notation \(m_i=|\mathcal A_i|\). For every adjacent pair \((i,j)\), the matrix \(\mathbf M_{ij}\in[0,1]^{m_i\times m_j}\) specifies player \(i\)'s payoff from the interaction with player \(j\).  If \(\boldsymbol e_{a_i}\) is the unit vector associated with action \(a_i\), then, writing \(N_i\) for the neighbors of \(i\), the normalized utility is given by
\begin{equation*}
 u_i(\boldsymbol a)
 =
 \frac1{|N_i|}
 \sum_{j\in N_i}
 \boldsymbol e_{a_i}^{\top}
 \mathbf M_{ij}
 \boldsymbol e_{a_j}.
\end{equation*}
We may assume that there are no isolated players.

Finally, in the (black-box) \emph{query model}, a query specifies a pure profile and returns every player's payoff at that profile.  The algorithm may be computationally unbounded between queries. A lower bound in this model is unconditional.

\subsection{Non-convexity of information-value-free CE and CCE}

To clarify the notion of an IVF(C)CE, we next provide an illustrative example; this will also show that the set of these equilibria can be nonconvex.

\begin{example}
    We consider the following game \(\mathcal{G}\), where \(\mathcal A_1 = \{\mathsf{U},\mathsf{D}\}\) and
\(\mathcal A_2=\{\mathsf{L},\mathsf{R}\}\):
\[
\renewcommand{\arraystretch}{1.2}
\begin{array}{c|cc}
             & \mathsf{L} & \mathsf{R} \\ \hline
\mathsf{U}   & \cellcolor{payoffgray}(1,1) & (0,0) \\
\mathsf{D}   & (0,0) & \cellcolor{payoffgray}(1,1)
\end{array}
\]
Let \(\mu\) be the point mass on \((\mathsf{U},\mathsf{L})\), and let \(\mu'\) be the point mass on \((\mathsf{D},\mathsf{R})\). Both profiles are pure Nash equilibria and thus
 information-value-free CE and CCE.

The distribution \(\frac12\mu+\frac12\mu'\) is a CE since the set of CE is convex. It is
therefore also a CCE. However, for every
player $i$,
\[
\E_{\boldsymbol a\sim \frac12\mu+\frac12\mu'}[u_i(\boldsymbol a)]= 1
\text{ and }
\E_{\boldsymbol a\sim \frac12\mu+\frac12\mu'}
   [u_i(a_i^*,\boldsymbol a_{-i})]=\frac12
\]
for every fixed action \(a_i^*\in\mathcal A_i\). Thus, \(\frac12\mu+\frac12\mu'\) is not information-value-free, proving that both sets are non-convex.
\end{example}

The previous example establishes the following.

\begin{proposition}[Example where recommendations are valuable]
There exists a game for which the set of information-value-free CEs or CCEs is not convex.
\end{proposition}

Another illustrative example is based on Shapley's game, which was analyzed by~\citet{Shapley63:Some} in order to prove that fictitious play does not converge in general-sum games.
\[
\renewcommand{\arraystretch}{1.2}
\begin{array}{c|ccc}
    & \mathsf{L} & \mathsf{M} & \mathsf{R} \\ \hline
\mathsf{U}
    & (0,0)
    & \cellcolor{payoffgray}(1,2)
    & \cellcolor{payoffgray}(2,1) \\
\mathsf{C}
    & \cellcolor{payoffgray}(2,1)
    & (0,0)
    & \cellcolor{payoffgray}(1,2) \\
\mathsf{D}
    & \cellcolor{payoffgray}(1,2)
    & \cellcolor{payoffgray}(2,1)
    & (0,0)
\end{array}
\]
The distribution that mixes uniformly over the outcomes in the gray cells above is a CE, but is not information-value-free.

We next point out an example of an IVF(C)CE that is not a Nash equilibrium.

\begin{example}[IVFCE that is not a Nash equilibrium]
Consider the following two-player game
\[
\renewcommand{\arraystretch}{1.2}
\begin{array}{c|cc}
             & \mathsf{L} & \mathsf{R} \\ \hline
\mathsf{U}   & \cellcolor{payoffgray}(1,0) & (0,0) \\
\mathsf{D}   & (0,0) & \cellcolor{payoffgray}(0,1)
\end{array}
\]
Let \(\mu\) be the point mass on \((\mathsf{U},\mathsf{L})\), and let \(\mu'\) be the point mass on \((\mathsf{D},\mathsf{R})\). Then it is not hard to verify that \(\frac12\mu+\frac12\mu'\) is an IVFCE, but is not a Nash equilibrium.
\end{example}

\section{Information-value-free correlated equilibrium}
\label{sec:ivfce}

In this section, we study the complexity of computing an \(\epsilon\)-\(\IVFCE\). We first establish an exponential payoff-query lower bound for general many-player games and then prove \(\PPAD\)-completeness in polymatrix games.

\subsection{Exponential query lower bound}

We will reduce from the black-box query lower bound of~\citet{Rubinstein16:Settling} for finding an approximate Brouwer fixed point, which is a strengthening the classic, earlier result of~\citet{Hirsch89:Exponential}.\footnote{We caution that~\citet{Rubinstein16:Settling} uses a normalized version of the Euclidean norm, whereas we use the usual, unnormalized definition.}

\begin{theorem}[Query lower bound for Brouwer fixed point; \citealp{Rubinstein16:Settling}]
\label{thm:black-box-brouwer}
There is a universal constant \(\alpha > 0\) such that every randomized algorithm with oracle access to a \((1/\alpha)\)-Lipschitz map \(f:[0,1]^n\to[0,1]^n\) requires at least \(2^{\Omega(n)}\) queries to $f$ in order to output, with probability at least \(2/3\), a point \(\boldsymbol x\) satisfying \(\|f(\boldsymbol x)-\boldsymbol x\|_2\le\alpha\sqrt n\).
\end{theorem}

\paragraph{The imitation game.}
Given a map \(f: [0,1]^n \rightarrow [0, 1]^n\) per~\cref{thm:black-box-brouwer}, we construct a $2n$-player game as follows. For each coordinate $i \in [n]$, we create two players \(A_i\) and \(B_i\). All players have the same action set \(Q_k \defeq \{0, 1/k, \ldots, 1\}\), which is a discretization of the $[0, 1]$ interval. At a pure profile \((\boldsymbol a,\boldsymbol b)\), we define the utility functions for all $i \in [n]$ as
\begin{equation}
 u_{A_i}(\boldsymbol a,\boldsymbol b)=1-(a_i-b_i)^2 \text{ and }
 u_{B_i}(\boldsymbol a,\boldsymbol b)=1-(b_i-f_i(\boldsymbol a))^2.
 \label{eq:brouwer-game}
\end{equation}

Here, player $A_i$ tries to predict player $B_i$'s action $b_i$, whereas player $B_i$ tries to predict $f_i(\boldsymbol{a})$. We can simulate each payoff query on $(\boldsymbol{a}, \boldsymbol{b})$ with one query at \( f(\boldsymbol a)\).

We now show that we can recover an approximate fixed point of $f$ from any approximate IVFCE. The key is the following lemma, showing that the requirement of correlated equilibrium and information-value-freeness together forces the correlated distribution to be nearly deterministic and close to an approximate Nash equilibrium. The high-level idea is this: given a joint distribution $\mu$ over action profiles, (1) $\mu$ being an approximate CE requires \emph{calibration $\approx$ posterior}; (2) $\mu$ being approximate information-value-free requires \emph{calibration $\approx$ prior}. This implies \emph{posterior $\approx$ prior} and thus the forecast is nearly deterministic.
\begin{lemma}
\label{lem:quadratic-no-information}
Let $u(x, y) = 1 - (x - y)^2: [0,1]\times[0,1] \rightarrow [0,1]$. For any joint distribution $\mu$ over $(x, y) \in Q_k \times [0,1]$, define $\epsilon_{\mathrm{CE}} = \mathbb{E}_{(x,y) \sim \mu}[\max_{x'\in Q_k} u(x',y)\mid x] - \mathbb{E}_{(x,y)\sim\mu}[u(x,y)]$ and $\epsilon_{\mathrm{IVF}} = \mathbb{E}_{(x,y)\sim\mu}[u(x,y)] - \max_{x'\in Q_k}\mathbb{E}_{(x,y)\sim\mu}[u(x',y)]$. Then we have
\begin{enumerate}[label=\arabic*.,ref=\arabic*]
    \item\label{item:quadratic-posterior-variance} $\operatorname{Var}(\E[y\mid x]) \le \epsilon_{\mathrm{CE}} + \epsilon_{\mathrm{IVF}} + 1/(4k^2)$;
    \item\label{item:quadratic-mean-square} $ \E[(x-\E[y])^2]\le 4\epsilon_{\mathrm{CE}} + 2\epsilon_{\mathrm{IVF}}+1/k^2$;
    \item\label{item:quadratic-bias-variance} $ |\E [x]-\E [y])|\le \sqrt{4\epsilon_{\mathrm{CE}} + 2\epsilon_{\mathrm{IVF}}+1/k^2}$ and $\operatorname{Var}(x) \le 4\epsilon_{\mathrm{CE}} + 2\epsilon_{\mathrm{IVF}}+1/k^2$;
\end{enumerate}
\end{lemma}

\begin{proof}
Since $\dist(\cdot ,Q_k) \le 1/(2k)$, we have
\begin{align*}
    \epsilon_{\mathrm{CE}} &=\mathbb{E}_{(x,y) \sim \mu}[\max_{x'\in Q_k} u(x',y)\mid x] - \mathbb{E}_{(x,y)\sim\mu}[u(x,y)]\\
    &= \E[(x- \E[y \mid x])^2] - \E[ \dist(\E[y \mid x], Q_k)^2] \\
    &\ge \E[(x- \E[y \mid x])^2] - \frac{1}{4k^2}.
\end{align*}
By definition of $\epsilon_{\mathrm{IVF}}$, we also have
\begin{align*}
    \epsilon_{\mathrm{CE}}+ \epsilon_{\mathrm{IVF}} &=\mathbb{E}_{(x,y) \sim \mu}[\max_{x'\in Q_k} u(x',y)\mid x] - \mathbb{E}_{(x,y)\sim\mu}[u(x,y)] \\
    &\quad + \mathbb{E}_{(x,y)\sim\mu}[u(x,y)] - \max_{x'\in Q_k}\mathbb{E}_{(x,y)\sim\mu}[u(x',y)]\\
    &= \mathbb{E}_{(x,y) \sim \mu}[\max_{x'\in Q_k} u(x',y)\mid x] - \max_{x'\in Q_k}\mathbb{E}_{(x,y)\sim\mu}[u(x',y)]\\
    &=\operatorname{Var}(\E[y\mid x]) - \E[ \dist(\E[y \mid x], Q_k)^2] + \E[ \dist(\E[y], Q_k)^2] \\
    &\ge \operatorname{Var}(\E[y\mid x]) - \frac{1}{4k^2}.
\end{align*}
This proves \Cref{item:quadratic-posterior-variance}.
Since $x - \E[y] = x - \E[y\mid x] + \E[y\mid x] - \E[y]$, we have
\begin{align*}
    \E[(x - \E[y])^2] &\le 2\E[(x - \E[y\mid x])^2] + 2\E[(\E[y\mid x] - \E[y])^2]\\
    &\le 2\E[(x - \E[y\mid x])^2] + 2\operatorname{Var}(\E[y\mid x])\\
    &\le 4\epsilon_{\mathrm{CE}} + 2\epsilon_{\mathrm{IVF}} + \frac{1}{k^2}.
\end{align*}
This proves \Cref{item:quadratic-mean-square}. Finally, by the bias-variance identity, we have $\E[(x - \E[y])^2] = \operatorname{Var}(x) + (\E[x] - \E[y])^2$. Combining this with \Cref{item:quadratic-mean-square} proves \Cref{item:quadratic-bias-variance}.
\end{proof}

We are now ready to establish an exponential query complexity lower bound for IFVCE, which we restate below.

\mainquery*

\begin{proof} Let \(\mu\) be an \(\epsilon\)-\(\IVFCE\) of the $2n$-player imitation game described above. We define $\boldsymbol{\bar{a}} = \E_{(\boldsymbol{a}, \boldsymbol{b}) \sim \mu}[\boldsymbol{a}]$ and $\boldsymbol{\bar{b}} = \E_{(\boldsymbol{a}, \boldsymbol{b}) \sim \mu}[\boldsymbol{b}]$. We define $\rho = \sqrt{6\epsilon+1/k^2}$. Applying \cref{lem:quadratic-no-information} to each player \(A_i\), with \(x=a_i\) and \(y=b_i\), gives \(|\bar{a}_i - \bar{b}_i|\le\rho\) and \(\operatorname{Var}(a_i)\le\rho^2\). Summing the coordinates, we get 
\begin{equation}
 \|\boldsymbol{\bar{a}}-\boldsymbol {\bar{b}}\|_2^2
 =\sum_{i=1}^n(\bar{a}_i-\bar{b}_i)^2\le n\rho^2
 \text{ and }
 \E[\|\boldsymbol{a}- \boldsymbol{\bar{a}}\|_2^2]
 =\sum_{i=1}^n\operatorname{Var}(a_i)\le n\rho^2.
 \label{eq:first-imitation-layer}
\end{equation}
Similarly, applying \Cref{lem:quadratic-no-information} to each player \(B_i\), with \(x=b_i\) and \(y=f_i(\boldsymbol a)\), gives
\[
 \|\boldsymbol{\bar{b}}-\E[f(\boldsymbol a)]\|_2^2
 \le n\rho^2.
\]
By Jensen's inequality and \eqref{eq:first-imitation-layer}, we get
\begin{align*}
    \|f(\boldsymbol{\bar{a}}) - \E [f(\boldsymbol a)]\|_2 &= \|\E [ f(\boldsymbol{\bar{a}})- f(\boldsymbol a)]\|_2 \\
    &\leq \E[ \| f(\boldsymbol{\bar{a}})- f(\boldsymbol a) \|_2 ] \\
 &\le (1/\alpha) \E[\|\boldsymbol{\bar{a}}-\boldsymbol a\|_2] \\
 &\le (1/\alpha) \sqrt n\ \rho,
\end{align*}
where we used the fact that $\|\cdot\|_2$ is convex and $f$ is $(1/\alpha)$-Lipschitz continuous. Moreover, by the triangle inequality,
\begin{equation*}
 \|f(\boldsymbol{\bar{a}})-\boldsymbol{\bar{a}}\|_2 \le  \|f(\boldsymbol{\bar{a}})-\E [f(\boldsymbol a)]\|_2 + \|\E [f(\boldsymbol a)] - \boldsymbol{\bar{b}}\|_2 + \|\boldsymbol{\bar{b}} - \boldsymbol{\bar{a}} \|_2
 \le(1/\alpha+2)\sqrt n\sqrt{6\epsilon+ \frac{1}{k^2}}.
\end{equation*}

We now choose the constants of the reduction as
\[k=\left\lceil\frac{8(1/\alpha+2)}\alpha\right\rceil
 \text{ and }
 \epsilon=\frac{\alpha^2}{384(1/\alpha+2)^2}.
\]
This parameter choice guarantees that
\[
\|f(\boldsymbol{\bar{a}})-\boldsymbol{\bar{a}}\|_2 \le (1/\alpha +2) \sqrt{n}\sqrt{6\epsilon + \frac{1}{k^2} } \le \alpha \sqrt{n}.
\]
As a result, the black-box lower bound in \cref{thm:black-box-brouwer} implies that finding an $\epsilon$-IVFCE in a $2n$-player game requires \(2^{\Omega(n)}\) payoff queries.
\end{proof}

\subsection{Computational complexity in polymatrix games} We now turn to the computational complexity of IVFCEs. We will reduce computing an approximate Nash equilibrium in constant-degree polymatrix games to computing an approximate IVFCE. This reduction hinges the following hardness result for computing an approximate Nash equilibrium~\citep{Rubinstein15:Inapproximability,Deligkas24:Pure}.

\begin{theorem}[\(\PPAD\)-hardness for Nash equilibrium in polymatrix games; \citealp{Rubinstein15:Inapproximability,Deligkas24:Pure}]
\label{thm:hard-polymatrix}
There is a universal constant \(\gamma\in(0,1]\) such that computing a \(\gamma\)-Nash equilibrium is \(\PPAD\)-hard even in bipartite polymatrix games of maximum degree three with two actions per player.
\end{theorem}

We recall that a $\gamma$-Nash equilibrium is a (product) strategy profile in which the (external) regret of each player is at most $\gamma$.

\paragraph{High-level idea.} Given an $n$-player $2$-action polymatrix game, we will construct an $n$-player game with the following action set: each player chooses a binary action, as in the original game, and a \emph{prediction} for each of its neighbor's action. A player's utility is a linear combination of the utility from the original game and a quadratic score corresponding with forecasting their neighbor's actions. As before, we will use the fact that CE calibrates the forecasts while information-value-freeness ensures that conditioning on the forecasts has little value. Together, these two properties imply that we can recover an approximate Nash equilibrium by taking the marginals of an approximate IVFCE.

\paragraph{The reduction.} More specifically, let \(\mathcal G\) be an $n$-player binary-action polymatrix game per~\cref{thm:hard-polymatrix}. This game is described through a bipartite degree-three graph \(([n],E)\), neighbor sets \(N_i\), binary action sets \(\mathcal A_i^{\mathcal G}=\{\mathsf{0},\mathsf{1}\}\), payoff matrices \(\mathbf M_{ij}^{\mathcal G}\in[0,1]^{2\times2}\) across edges $(i, j) \in E$, and normalized utilities
\[
 u_i^{\mathcal G}(\boldsymbol a)
 =\frac1{|N_i|}\sum_{j\in N_i}
 \boldsymbol e_{a_i}^{\top}\mathbf M_{ij}^{\mathcal G}\boldsymbol e_{a_j}.
\]
We consider an integer \(k\ge1\) and a weight \(\lambda\in(0,1)\), whose values will be specified later. As before, we discretize the interval $[0, 1]$ by taking  \(Q_k \defeq \{0,1/k,\ldots,1\}\). Further, we denote by \(B(q, a)=1-(q - a)^2\) the quadratic score.\footnote{We use the notation $B(\cdot, \cdot)$ as an allusion to the \emph{Brier score}, although we caution that the Brier score is typically thought of as a cost function.} We construct a game \(\mathcal H\) on the same graph. Player \(i\)'s action set is $\mathcal A_i^{\mathcal H}
 \defeq \mathcal A_i^{\mathcal G}\times Q_k^{N_i}$. We will the shorthand notation $y_i = (a_i,(q_{ij})_{j\in N_i} )$, where \(a_i\in\{\mathsf{0},\mathsf{1}\}\) is the source action and \(q_{ij}\) forecasts the probability that neighbor \(j\)'s source action equals \(\mathsf{1}\). Thus, each player $i$ has at most \(m_i^{\mathcal H}=|\mathcal A_i^{\mathcal H}|=2(k+1)^{|N_i|}\le2(k+1)^3\) actions. Finally, for every ordered adjacent pair \((i,j)\), we define the payoff matrix of \(\mathcal H\) entrywise by
\[
 [\mathbf M_{ij}^{\mathcal H}]_{y_i,y_j}
 =\lambda\boldsymbol e_{a_i}^{\top}\mathbf M_{ij}^{\mathcal G}\boldsymbol e_{a_j}
 +(1-\lambda)B(q_{ij},a_j),
\]
so that for a strategy profile \(\boldsymbol y=(y_i)_{i\in[n]}\) and its source-action coordinates \(\boldsymbol a=(a_i)_{i\in[n]}\), the utility of player $i$ reads
\begin{equation}
 u_i^{\mathcal H}(\boldsymbol y)
 =\frac1{|N_i|}\sum_{j\in N_i}[\mathbf M_{ij}^{\mathcal H}]_{y_i,y_j}
 =\lambda u_i^{\mathcal G}(\boldsymbol a)
 +\frac{1-\lambda}{|N_i|}\sum_{j\in N_i}B(q_{ij},a_j).
 \label{eq:forecast-payoff}
\end{equation}
In particular, $\mathcal{H}$ is indeed an $n$-player constant-action constant-degree polymatrix game. Our goal is to establish the following theorem, which we restate from the introduction.

\mainppad*

\begin{proof}
Membership in $\PPAD$ follows from the fact that an exact mixed Nash equilibrium of a polymatrix game belongs to \(\PPAD\)~\citep{FilosRatsikas24:PPAD}. For the hardness, let \(\mu\) be an \(\epsilon\)-\(\IVFCE\) of the constructed game \(\mathcal H\). We are interested in drawing an action profile \(\boldsymbol y\sim\mu\). Let \(\boldsymbol a=(a_i)_{i\in[n]}\) be the corresponding source actions. For every player $i$ and every adjacent pair $(i,j) \in E$, we define $p_j=\Prb[a_j=\mathsf{1}]$,  $P_{ij}=\Prb[a_j=\mathsf{1}\mid y_i]$, and $\boldsymbol x_j=(1-p_j,p_j)$. By the law of conditional probability, we have \(\E P_{ij}= p_j \). In this context, for a player \(i\), we define the following three scores:
\begin{itemize}
    \item[1.]  \emph{Forecasting score:} $\frac1{|N_i|}\sum_{j\in N_i}\E_\mu[B(q_{ij},a_j)]$, which is the score obtained by following the recommendations.
    \item[2.] \emph{Posterior score:} $\frac1{|N_i|}\sum_{j\in N_i}\E[ \max_{q_{ij}' \in Q_k} B(q_{ij}', a_j)\mid y_i]$, which is the score obtained by choosing the conditionally optimal report upon observing a recommendation.
    \item[3.] \emph{Prior score:} $\frac1{|N_i|}\sum_{j\in N_i} \max_{q_{ij}' \in Q_k} \E_{\mu} [B(q_{ij}', a_j)]$, which is the score by choosing the optimal fixed report without observing the recommendations.
\end{itemize}
The goal now is to apply~\Cref{lem:quadratic-no-information}. However, that lemma concerns the imitation game, which contains only the component of the utility corresponding to the quadratic score. This issue can be addressed by taking $\lambda$ to be sufficiently small, so that the quadratic score dominates the utility of each player. 

After observing \(y_i\), player \(i\) may leave \(a_i\) unchanged and replace every report \(q_{ij}\) by a conditionally optimal grid point. This is a valid strategy deviation which alters only the forecasting component. Since $\mu$ is an $\epsilon$-CE, it follows that 
\[
    \frac1{|N_i|}\sum_{j\in N_i}\E[ \max_{q_{ij}' \in Q_k} B(q_{ij}', a_j)\mid y_i] \leq \frac1{|N_i|}\sum_{j\in N_i}\E_\mu[B(q_{ij},a_j)] + \frac{\epsilon}{1 - \lambda}.
\]
Moreover, since the utilities in the original game lie in $[0,1]$, player $i$'s best fixed action has value at most 
\[
    \lambda + (1- \lambda) \frac1{|N_i|}\sum_{j\in N_i} \max_{q_{ij}' \in Q_k} \E_{\mu} [B(q_{ij}', a_j)],
\]
while following the recommendation guarantees value at least 
\[
  (1-\lambda) \frac1{|N_i|}\sum_{j\in N_i}\E_\mu[B(q_{ij},a_j)]
\]
Since $\mu$ is $\epsilon$-information-value-free, we get
\[
   (1-\lambda) \frac1{|N_i|}\sum_{j\in N_i}\E_\mu[B(q_{ij},a_j)] \leq \lambda + (1- \lambda) \frac1{|N_i|}\sum_{j\in N_i} \max_{q_{ij}' \in Q_k} \E_{\mu} [B(q_{ij}', a_j)] + \epsilon.
\]
Equivalently,
\[
    \frac1{|N_i|}\sum_{j\in N_i}\E_\mu[B(q_{ij},a_j)] \leq \frac1{|N_i|}\sum_{j\in N_i} \max_{q_{ij}' \in Q_k} \E_{\mu} [B(q_{ij}', a_j)] + \frac{\epsilon + \lambda}{1 - \lambda}.
\]
As a result, applying~\Cref{lem:quadratic-no-information} yields
\begin{equation}
 \frac1{|N_i|}\sum_{j\in N_i}\E[(P_{ij}-p_j)^2]
 \le
 \frac{2\epsilon+\lambda}{1-\lambda}+\frac1{4k^2}.
 \label{eq:information-bound}
\end{equation}

We now argue that the marginals of $\mu$, projected onto the source actions $\{\boldsymbol x_i = (1-p_i, p_i)\}_{i \in [n]}$, constitute an approximate Nash equilibrium of the original game $\mathcal{G}$. Conditional on \(y_i\), we write
\(r_i(a_i \mid y_i) \defeq \E[u_i^{\mathcal G}(a_i ,\boldsymbol a_{-i})\mid y_i]\) for \(a_i \in\{\mathsf{0},\mathsf{1}\}\). A deviation in \(\mathcal H\) may replace only the source action and leave all reports unchanged. Choosing a conditionally best source action and applying CE obedience gives
\begin{equation}
 \E\left[\max_{a_i' \in \{\mathsf{0},\mathsf{1}\} } r_i(a_i'\mid y_i)-r_i(a_i\mid y_i) \right]
 \le\frac{\epsilon}{\lambda}.
 \label{eq:conditional-obedience}
\end{equation}
Replacing the posterior law of
\(a_j\) by its marginal changes the expected payoff on edge \((i,j)\)
by
\[
\left|
(P_{ij}-p_j)
\left(
\mathbf M_{ij}^{\mathcal G}(a_i,\mathsf{1})
-\mathbf M_{ij}^{\mathcal G}(a_i,\mathsf{0})
\right)
\right|
\le |P_{ij}-p_j|,
\]
where the inequality follows because the entries of every \(\mathbf M_{ij}^{\mathcal G}\) lie in \([0,1]\). Thus, for either source action \(a_i \in \{\mathsf{0},\mathsf{1}\} \),
\[
 |r_i(a_i \mid y_i)- u_i^{\mathcal G}(a_i,\boldsymbol x_{-i}) |
 \le \frac1{|N_i|}\sum_{j\in N_i}|P_{ij}-p_j|.
\]
Applying the above inequality twice gives
\[
 \max_{a_i' \in \{\mathsf{0},\mathsf{1}\} } u_i^{\mathcal G}(a'_i, \boldsymbol x_{-i}) -u_i^{\mathcal G}(a_i,\boldsymbol x_{-i})
 \le \max_{a_i' \in \{\mathsf{0},\mathsf{1}\} } r_i(a_i'\mid y_i)-r_i(a_i\mid y_i)
 +\frac2{|N_i|}\sum_{j\in N_i}|P_{ij}-p_j|.
\]
Taking expectations, using \eqref{eq:conditional-obedience}, and applying Cauchy--Schwarz with \eqref{eq:information-bound} yield
\begin{equation}
 \max_{a_i \in \{\mathsf{0},\mathsf{1}\}}u_i^{\mathcal G}(a_i, \boldsymbol x_{-i})
 -u_i^{\mathcal G}(\boldsymbol x_i,\boldsymbol x_{-i})
 \le \frac{\epsilon}{\lambda}
 +2\sqrt{\frac{2\epsilon+\lambda}{1-\lambda}+\frac1{4k^2}}.
 \label{eq:extraction-bound}
\end{equation}

Now, let \(\gamma\) be the constant in \cref{thm:hard-polymatrix}. We choose the parameters as $\lambda=(\gamma/16)^2$, $k=\lceil16/\gamma\rceil$, and $\epsilon=\gamma\lambda/4$, which are universal constants. The constructed game has at most \(m=2(k+1)^3\) actions per player, which is also a universal constant. Moreover, we have \(\epsilon/\lambda=\gamma/4\), \(\gamma\le1\) implies
\((2\epsilon+\lambda)/(1-\lambda)\le3\lambda\), while the choice of \(k\) gives \(1/(4k^2)\le\gamma^2/1024\). Substituting these bounds into \eqref{eq:extraction-bound} yields
\[
\max_{a_i \in\{\mathsf{0},\mathsf{1}\}}
 \left[u_i^{\mathcal G}(a_i,\boldsymbol x_{-i})
 -u_i^{\mathcal G}(\boldsymbol x_i,\boldsymbol x_{-i})\right]
 \le\frac\gamma4+\frac{\sqrt{13}}{16}\gamma<\gamma.
\]
Thus, any algorithm for computing an \(\epsilon\)-\(\IVFCE\) of \(\mathcal H\) would give a \(\gamma\)-approximate Nash equilibrium of \(\mathcal G\), and \(\PPAD\)-completeness follows.
\end{proof}
\section{Information-value-free coarse correlated equilibrium}
\label{sec:ivfcce}

For information-value-free CCEs, we first point out algorithms such as regret matching or FTRL yield an FPTAS. When the precision $\epsilon$ is exponentially small, we show that the previous algorithm based on no-regret dynamics can be turned into a $\CLS$ membership argument. However, the problem is at least as hard as solving LCPs with a P-matrix, so unlikely to admit a polynomial-time algorithm.

\subsection{An FPTAS via no-regret learning}
\label{subsec:ivfcce-fptas}

For mixed-strategy profiles \(\boldsymbol x^{(1)},\ldots,\boldsymbol x^{(T)}\), we let
\(\mu^{(T)} =T^{-1}\sum_{t=1}^T\bigotimes_{i\in[n]}\boldsymbol x_i^{(t)}\) be the average correlated distribution. In this case, we can write the external regret, which controls the CCE gap, as
\begin{equation}
 \barreg_i(\mu^{(T)})
 =\frac1T\max_{a_i'\in\mathcal A_i}\sum_{t=1}^T
 \left[u_{i}(a_i',\boldsymbol x_{-i}^{(t)})
 -u_i(\boldsymbol x^{(t)} )\right].
 \label{eq:learning-regret-identity}
\end{equation}

\subsubsection{Information-value-free no-regret dynamics}

We first consider the \emph{regret matching} algorithm~\citep{Hart00:Simple}, with a suitable tie-breaking rule when the regrets are nonpositive.

\paragraph{Regret matching.}
For each player \(i\in[n]\) and action \(a_i \in \mathcal{A}_i \), we define the cumulative regret up to round \(t\) as
\[
 r_{a_i}^{(t)}
 \defeq \sum_{\tau=1}^t
 \left[u_{i}(a_i, \boldsymbol x_{-i}^{(\tau)})
 -u_i(\boldsymbol x^{(\tau)})\right].
\]
We write \(\boldsymbol r_i^{(t)}=(r_{a_i}^{(t)})_{a_i\in\mathcal A_i}\) for the corresponding regret vector. By convention, we take $r_{a_i}^{(0)}=0$ for all actions $a_i \in \mathcal{A}_i$. On round \(t\), we define regret matching as
\begin{equation}
 x_{a_i}^{(t) }=
 \begin{cases}
 \displaystyle
 \frac{\pos{r_{a_i}^{(t-1)}}}
 {\sum_{a_i'\in\mathcal A_i}\pos{r_{a_i'}^{(t-1)}}}
 &\text{if }\sum_{a_i'\in\mathcal A_i} \pos{r_{a_i'}^{(t-1)}}>0,\\[2.5mm]
 \displaystyle
 \frac{\ind{a_i\in\argmaxop_{a_i'\in\mathcal A_i} r_{a_i'}^{(t-1)}}}
 {|\argmaxop_{a_i'\in\mathcal A_i} r_{a_i'}^{(t-1)}|}
 &\text{otherwise.}
 \end{cases}
 \label{eq:mixed-zero-safe-rm}
\end{equation}
In the latter case, it is common to allow the strategy to be arbitrary in regret matching, but here we need handle it more carefully. The proof of the upper bound is standard, but we include here because the underlying potential will play a key role in the $\PLS$ membership result.

\begin{proposition}[Regret matching]
\label{prop:zero-safe-rm}
At every horizon \(T\), regret matching~\eqref{eq:mixed-zero-safe-rm} satisfies
\[
 0\le\barreg_i(\mu^{(T)})\le\sqrt{m_i/T}.
\]
\end{proposition}

\begin{proof}
We consider a player \(i \in [n]\). For the lower bound, we proceed by induction. The base case is \(\boldsymbol r_i^{(0)} = \boldsymbol{0} \). Suppose \(\max_{a_i\in\mathcal A_i} r_{a_i}^{(t-1)}\ge0\). We first prove that $\sum_{a_i \in \mathcal{A}_i} x_{a_i}^{(t)} r_{a_i}^{(t-1)} \ge 0$. In the first branch of
\eqref{eq:mixed-zero-safe-rm},
\[
 \sum_{a_i\in\mathcal A_i} x_{a_i}^{(t)}r_{a_i}^{(t-1)}
 =\frac{\sum_{a_i\in\mathcal A_i} \pos{r_{a_i}^{(t-1)}}^2}
 {\sum_{a_i\in\mathcal A_i}\pos{r_{a_i}^{(t-1)}}}>0.
\]
In the second branch, the inductive hypothesis forces the maximum cumulative gain to equal zero, so the same weighted sum equals zero. Moreover, if
\(g_{a_i}^{(t)}=u_{i}(a_i, \boldsymbol x_{-i}^{(t)})-u_i(\boldsymbol x_i^{(t)},\boldsymbol x_{-i}^{(t)})\), then
\(\sum_{a_i\in\mathcal A_i} x_{a_i}^{(t)}g_{a_i}^{(t)}=0\). Hence,
\[
 \sum_{a_i\in\mathcal A_i} x_{a_i}^{(t)}r_{a_i}^{(t)}
 = \sum_{a_i\in\mathcal A_i} x_{a_i}^{(t)}(r_{a_i}^{(t-1)} + g_{a_i}^{(t)})\ge0,
\]
which implies \(\max_{a_i\in\mathcal A_i} r_{a_i}^{(t)}\ge0\). This completes the induction.

For the upper bound, we consider the potential
\[
 \Phi_i(\boldsymbol r_i)
 \defeq\sum_{a_i\in\mathcal A_i}\pos{r_{a_i}}^2.
\]
Regret matching satisfies, by definition, the following orthogonality:
\[
 \sum_{a_i\in\mathcal A_i}\pos{r_{a_i}^{(t-1)}}g_{a_i}^{(t)}=0.
\]
Using the inequality \(\pos{r+s}^2\le\pos r^2+2\pos r\,s+s^2\) and the fact that \(|g_{a_i}^{(t)}|\le1\) (due to normalization), we obtain
\begin{align*}
 \Phi_i(\boldsymbol r_i^{(t)})
 &\le\Phi_i(\boldsymbol r_i^{(t-1)})
 +2\sum_{a_i\in\mathcal A_i}\pos{r_{a_i}^{(t-1)}}g_{a_i}^{(t)}
 +\sum_{a_i\in\mathcal A_i}(g_{a_i}^{(t)})^2\\
 &\le\Phi_i(\boldsymbol r_i^{(t-1)})+m_i.
\end{align*}
Since \(\Phi_i(\boldsymbol r_i^{(0)})=0\), the telescopic summation gives \(\Phi_i(\boldsymbol r_i^{(T)})\le m_iT\). Combining this with the lower bound and~\eqref{eq:learning-regret-identity} yields
\[
 0\le\barreg_i(\mu^{(T)})
 =\frac1T\max_{a_i\in\mathcal A_i}r_{a_i}^{(T)}
 \le\frac{\sqrt{\Phi_i(\boldsymbol r_i^{(T)})}}{T}
 \le\sqrt{m_i/T}.
\]
\end{proof}

In the local-search construction below, we will use the sum of these per-player potentials.

\paragraph{Nonnegativity of FTRL.} This nonnegativity property in fact holds for any follow-the-regularized-leader (FTRL) algorithm, so long as the regularizer remains fixed.\footnote{When the time horizon is not known in advance, one typically takes the learning rate to be time-dependent.}

\begin{restatable}[FTRL has nonnegative regret for fixed regularizers; \citealp{Guzman21:Best}]{proposition}{ftrlnonnegative}
\label{prop:ftrl-fptas}
For every horizon \(T\) and a fixed learning rate \(\eta_i>0\), FTRL satisfies $\barreg_i(\mu^{(T)}) \geq 0$.
\end{restatable}

For completeness, we provide the simple proof in~\Cref{app:fixed-step-ftrl}. We thus arrive at the following consequence.

\begin{corollary}[FPTAS]
\label{thm:ivfcce-fptas}
Consider a game with the polynomial expectation property, and let \(d=\sum_{i\in[n]}|\mathcal A_i|\). For every rational \(\epsilon\in(0,1]\), a deterministic algorithm running in time polynomial in the input length and \(1/\epsilon\) computes an $(\epsilon, 0)$-\(\IVFCCE\), with support at most \(T=\lceil d/\epsilon^2\rceil\).
\end{corollary}

We will next leverage no-regret dynamics to establish a polynomial local search formulation.

\subsection{Complexity at exponential precision}
\label{subsec:ivfcce-binary}

When the desired accuracy is encoded in binary, the previous algorithm can take exponential time. This section examines the complexity of IVFCCEs in this high-accuracy regime. Our first result is computing an approximate IVFCCE is in $\PLS$ and thus in $\CLS = \PPAD \cap \PLS$ (\Cref{thm:ivfcce-binary}). We then show that computation of IVFCCE is at least as hard as the P-matrix Linear complementarity problem (LCP) (\Cref{theorem:PLCP}) and thus is unlikely to admit a polynomial-time algorithm.

\subsubsection{CLS membership}
\label{sec:CLS}

We begin with the positive result.

\begin{theorem}
\label{thm:ivfcce-binary}
Consider a game with the polynomial expectation property. Given a rational \(\epsilon \in \Q_{> 0} \), computing an \(\epsilon\)-\(\IVFCCE\) belongs to $\CLS = \PPAD \cap \PLS$.
\end{theorem}
Membership in $\PPAD$ follows from existing results concerning Nash equilibria~\citep{Papadimitriou23:Computational}. Moreover, on account of the fact that $\CLS = \PPAD \cap \PLS$~\citep{Fearnley23:Complexity}, it suffices to establish membership in $\PLS$~\citep{Johnson88:How}. In particular, we need to establish that there is a discrete state space in which each state has polynomial encoding length, a polynomial-time algorithm for determining the next state from any given state, and a polynomial-time computable potential that strictly increases upon each transition unless the current state is a local optimum.

In what follows, if $L$ represents the number of bits needed to describe the input, we recall that we assume that there is an integer \(\payoffdenom\ge1\) encoded with \(\poly(L)\) bits such that \(u_i(\boldsymbol a')\in \payoffdenom^{-1} \Z \cap[0,1]\) for every player \(i\) and action profile \(\boldsymbol a'\). For each player \(i\), action \(a_i\in\mathcal A_i\), and action profile \(\boldsymbol a'\), the scaled payoff difference \(\payoffdenom[u_i(a_i,\boldsymbol a'_{-i})-u_i(\boldsymbol a')]\) is an integer in \([-\payoffdenom,\payoffdenom]\). In accordance with the proof in~\Cref{prop:zero-safe-rm}, for a scaled cumulative regret vector \(\boldsymbol r\in\Z^d\), we define the potential
\begin{equation}
 \Phi(\boldsymbol r)
 =\sum_{i\in[n]}\sum_{a_i\in\mathcal A_i}\pos{r_{i,a_i}}^2.
 \label{eq:positive-potential}
\end{equation}
Assuming the invariance \(\max_{a_i\in\mathcal A_i} r_{i,a_i}\ge0\) for every player, we use the shorthand notation \(\boldsymbol x_i(\boldsymbol r)\) for the (mixed) strategy prescribed by regret matching~\eqref{eq:mixed-zero-safe-rm} with cumulative regret coordinates \(r_{i,a_i}\). Due to scale invariance, scaling all regrets by a common positive factor \(\payoffdenom\) does not change the prescribed strategy.

The role of the next lemma (\Cref{lem:zero-drift-selection}) is to convert the mixed strategy obtained by regret matching into a pure strategy. This is essential for the $\PLS$ membership result, for otherwise the bit complexity of the strategy could blow up as the dynamics go along. To do so, we let \(\boldsymbol x(\boldsymbol r)=(\boldsymbol x_i(\boldsymbol r))_i\). We want to analyze what happens when an action profile is drawn from $\boldsymbol x(\boldsymbol r)$. It will be helpful to define the term
\begin{equation}
 \ell_{\boldsymbol r}(\boldsymbol a')
 =2\payoffdenom\sum_{i\in[n]}\sum_{a_i\in\mathcal A_i}
 \pos{r_{i,a_i}}\bigl[u_i(a_i,\boldsymbol a'_{-i})-u_i(\boldsymbol a')\bigr].
 \label{eq:linear-drift}
\end{equation}
(Above, the multiplication with $B$ only serves to make the expression integral.) To avoid randomization, we observe that there is a way to deterministically select a pure strategy profile while maintaining the orthogonality invariance, which is central in the regret analysis of regret matching.

\begin{lemma}[Deterministic purification]
\label{lem:zero-drift-selection}
For every regret vector \(\boldsymbol r\),
\(\E_{\boldsymbol a \sim\boldsymbol x(\boldsymbol r)}[\ell_{\boldsymbol r}(\boldsymbol a)]=0\). Moreover, one can compute in polynomial time a pure profile \(\boldsymbol a'(\boldsymbol r)\) in the support of \(\boldsymbol x(\boldsymbol r)\) such that
\(\ell_{\boldsymbol r}(\boldsymbol a'(\boldsymbol r))\le0\).
\end{lemma}

\begin{proof}
We fix a player \(i\). If all coordinates of the regret vector are nonpositive, the overall contribution is zero. Otherwise, independence and \eqref{eq:linear-drift} show that the expected contribution (divided by \(2\payoffdenom\)) is again zero by the definition of the first branch in~\eqref{eq:mixed-zero-safe-rm}. Summing over players proves the claim.

For the deterministic selection, we consider the players in a fixed order. Before fixing player \(i\), the current conditional expectation is a weighted average of the values obtained by fixing each action in its support. We choose, for example, the lexicographically first action that does not increase this expectation. The product-expectation oracle evaluates these conditional expectations exactly, with already processed players selecting pure strategies. After all players are fixed, the resulting pure profile satisfies \(\ell_{\boldsymbol r}(\boldsymbol a'(\boldsymbol r)) \le \E_{\boldsymbol a \sim\boldsymbol x(\boldsymbol r)}[\ell_{\boldsymbol r}(\boldsymbol a)] = 0 \).
\end{proof}

The next lemma tracks the one-step update of the regret vector and the potential function under the above deterministic selection. A minor detail here is that the regret vector is scaled by $B$ to guarantee integrality; this not affect the actual mixed strategy due to scale invariance.

\begin{lemma}[One-step invariants]
\label{lem:zero-safe-step}
Suppose \(\max_{a_i\in\mathcal A_i} r_{i,a_i}\ge0\) for all \(i\), \(\Phi(\boldsymbol r)\le d\payoffdenom^2t\), and \(|r_{i,a_i}|\le t\payoffdenom\). Let \(\boldsymbol a'=\boldsymbol a'(\boldsymbol r)\) be given by \cref{lem:zero-drift-selection} and set $r'_{i,a_i}=r_{i,a_i}+\payoffdenom [u_i(a_i,\boldsymbol a'_{-i})-u_i(\boldsymbol a') ]$ for every $i$ and $a_i \in \mathcal{A}_i$. Then $\max_{a_i\in\mathcal A_i} r'_{i,a_i} \ge 0$ for all $i \in [n]$, $\Phi(\boldsymbol r')\le d\payoffdenom^2(t+1)$, and $|r'_{i,a_i}|\le(t+1)\payoffdenom$.
\end{lemma}

\begin{proof}
Because \(a_i'\in\supp(\boldsymbol x_i(\boldsymbol r))\), definition~\eqref{eq:mixed-zero-safe-rm} gives \(r_{i,a_i'}\ge0\). Deviating from \(a_i'\) to itself has zero gain, so
\(r'_{i,a_i'}=r_{i,a_i'} \ge 0\).

For the potential upper bound, we use the inequality
\(\pos{r+s}^2\le\pos r^2+2\pos r\,s+s^2\) to every coordinate. There are \(d\) increments, each of magnitude at most \(\payoffdenom\), and their linear terms sum to
\(\ell_{\boldsymbol r}(\boldsymbol a')\le0\). Therefore,
\begin{align*}
 \Phi(\boldsymbol r') \le\Phi(\boldsymbol r)+\ell_{\boldsymbol r}(\boldsymbol a')
 +\payoffdenom^2\sum_{i\in[n]}\sum_{a_i\in\mathcal A_i}\bigl[u_i(a_i,\boldsymbol a'_{-i})-u_i(\boldsymbol a')\bigr]^2 \le d\payoffdenom^2(t+1).
\end{align*}
The final bound follows from the triangle inequality.
\end{proof}

Inductively applying~\cref{lem:zero-safe-step}, we arrive at the following invariances.

\begin{proposition}
\label{prop:deterministic-zero-safe-trajectory}
Starting from \(\boldsymbol r^{(0)}=\boldsymbol0\), we iteratively choose
\(\boldsymbol a^{\prime (t+1)}=\boldsymbol a'(\boldsymbol r^{(t)})\) and set $r_{i,a_i}^{(t+1)}=r_{i,a_i}^{(t)}
 +\payoffdenom [u_i(a_i,\boldsymbol a_{-i}^{\prime (t+1)})-u_i(\boldsymbol a^{\prime (t+1)}) ]$ for every $i$ and $a_i\in\mathcal A_i$. For every \(t\ge1\), we have $\max_{a_i\in\mathcal A_i} r_{i,a_i}^{(t)} \ge 0$ for all $i$, $\Phi(\boldsymbol r^{(t)})\le d\payoffdenom^2t$, and $|r_{i,a_i}^{(t)}|\le t\payoffdenom$.
\end{proposition}

Crucially, each round in the above procedure takes polynomial time. Now, for a problem to be in $\PLS$, we cannot output an exponentially long trajectory. Instead, we need a way to express the output succinctly. To do so, we use the following constructive form of Carath\'eodory's theorem.

\begin{lemma}[Trajectory compression]
\label{lem:rational-compression}
Given rational vectors \(\boldsymbol v^{(1)},\ldots,\boldsymbol v^{(K)}\in\Q^d\) and rational weights \(\boldsymbol\lambda\in\Delta([K])\), one can compute in polynomial time weights \(\widehat{\boldsymbol\lambda}\) supported on at most \(d+1\) indices such that
\(\sum_r\widehat\lambda_r\boldsymbol v^{(r)}=\sum_r\lambda_r\boldsymbol v^{(r)}\). Moreover, if the vectors are integral and bounded by \(\payoffdenom\), and their average is \(\boldsymbol r/t\) for \(\boldsymbol r\in[-t\payoffdenom,t\payoffdenom]^d\cap\Z^d\), each nonzero weight can be written with numerator and denominator of absolute value at most $D^{(t)} \defeq t(d+1)^{d+1}(\payoffdenom+1)^{d+1}.$
\end{lemma}

\begin{proof}
Let \(\boldsymbol v=\sum_r\lambda_r\boldsymbol v^{(r)}\). We consider the polytope of nonnegative weights \(\boldsymbol\theta\) satisfying
\(\sum_r\theta_r=1\) and \(\sum_r\theta_r\boldsymbol v^{(r)}=\boldsymbol v\). A rational vertex of this polytope can be found in polynomial time. If a feasible vertex was supported on more than \(d+1\) vectors, the corresponding augmented columns \((1,\boldsymbol v^{(r)})\in\Q^{d+1}\) would be linearly dependent, contradicting extremality.

On an affinely independent support of size \(r\le d+1\), we choose \(r\) linearly independent rows of the augmented columns \((1,\boldsymbol v^{(r)})\). The positive weights are the unique solution of the resulting \(r\)-by-\(r\) system with the corresponding entries of \((1,\boldsymbol r/t)\) on the right. Applying Cramer's rule and using Hadamard's inequality, it follows that every numerator and denominator is bounded by
\(t\,r^{r/2}(\payoffdenom+1)^r\le D^{(t)}\), as claimed.
\end{proof}

Now, if a sparse distribution \(\mu^{(t)}\) satisfies
\[
 t\payoffdenom\E_{\boldsymbol a'\sim\mu^{(t)}}
 \bigl[u_i(a_i,\boldsymbol a'_{-i})-u_i(\boldsymbol a')\bigr]=r_{i,a_i}^{(t)}
\]
for every $i$ and $a_i \in \mathcal{A}_i$, we form the distribution
\[
 \widetilde{\mu}^{(t+1)}
 =\frac{t}{t+1}\mu^{(t)}
 +\frac1{t+1} \delta_{\boldsymbol a^{\prime (t+1)}},
\]
where $\delta_{\boldsymbol a^{\prime (t+1)}}$ is the singleton supported on $a^{\prime (t+1)}$, and apply \cref{lem:rational-compression} to the vectors $( \payoffdenom\bigl[u_i(a_i,\boldsymbol a'_{-i})-u_i(\boldsymbol a')\bigr] )_{i,a_i}$ associated with its support profiles \(\boldsymbol a'\).

\begin{proposition}[Compressed trajectory]
\label{prop:compressed-trajectory}
For every \(t\ge1\), the preceding deterministic procedure constructs \(\boldsymbol r^{(t)}\in\Z^d\) and an explicit rational distribution \(\mu^{(t)}\) such that
\begin{equation}
 \begin{gathered}
 t\payoffdenom\E_{\boldsymbol a'\sim\mu^{(t)}}
 \bigl[u_i(a_i,\boldsymbol a'_{-i})-u_i(\boldsymbol a')\bigr]=r_{i,a_i}^{(t)}
 \quad\text{for every }i\text{ and }a_i\in\mathcal A_i,\\
 \max_{a_i\in\mathcal A_i} r_{i,a_i}^{(t)}\ge0\ \ \forall i,
 \quad
 \Phi(\boldsymbol r^{(t)})\le d\payoffdenom^2t,
 \quad
 |\supp(\mu^{(t)})|\le d+1.
 \end{gathered}
 \label{eq:trajectory-invariants}
\end{equation}
Every support weight has numerator and denominator bounded by \(D^{(t)}\) (\Cref{lem:rational-compression}), and each step can be implemented in polynomial time.
\end{proposition}

\begin{proof}
The base of the induction is when the underlying distribution is the singleton supported on \(\boldsymbol a^{\prime (1)}\). We will use the transition specified in~\Cref{prop:deterministic-zero-safe-trajectory}. Before applying the compression of~\Cref{lem:rational-compression}, for every \(i\) and \(a_i\in\mathcal A_i\), we have
\begin{align*}
 (t+1)\payoffdenom&\E_{\boldsymbol a'\sim\widetilde\mu^{(t+1)}}
 \bigl[u_i(a_i,\boldsymbol a'_{-i})-u_i(\boldsymbol a')\bigr]\\
 &\quad=t\payoffdenom\E_{\boldsymbol a'\sim\mu^{(t)}}
 \bigl[u_i(a_i,\boldsymbol a'_{-i})-u_i(\boldsymbol a')\bigr]
 +\payoffdenom\bigl[u_i(a_i,\boldsymbol a_{-i}^{\prime (t+1)})-u_i(\boldsymbol a^{\prime (t+1)})\bigr] = r_{i,a_i}^{(t+1)}.
\end{align*}
Applying the compression of~\Cref{lem:rational-compression} preserves these expectations, reduces the support of the distribution to at most \(d+1\), and the mixing weights over the support satisfy the claimed bound due to~\Cref{lem:rational-compression}.
\end{proof}

Now, let \(T=\lceil d/\epsilon^2\rceil\). Although \(T\) may be exponential, the counter \(t\), the integer vector \(\boldsymbol r^{(t)}\), and the compressed distribution have encodings polynomial in \(\log T\).

A \emph{valid state} is the initial state \(s^{(0)}=(0,\boldsymbol0,\varnothing)\), or a tuple \(s=(t,\boldsymbol r,\mu)\) with \(1\le t\le T\) satisfying the following properties.

\begin{enumerate}[label=\textnormal{(\arabic*)},leftmargin=2.8em]
 \item \(\boldsymbol r\in\Z^d\) and \(|r_{i,a_i}|\le t\payoffdenom\);
 \item \(\mu\) is an explicit distribution on at most \(d+1\) profiles, whose positive weights have numerators and denominators bounded by \(D^{(t)}\) (per~\Cref{lem:rational-compression});
 \item \(t\payoffdenom\E_{\boldsymbol a'\sim\mu}[u_i(a_i,\boldsymbol a'_{-i})-u_i(\boldsymbol a')]=r_{i,a_i}\) for every \(i\) and \(a_i\in\mathcal A_i\); and
 \item \(\max_{a_i\in\mathcal A_i} r_{i,a_i}\ge0\) for every player and \(\Phi(\boldsymbol r)\le d\payoffdenom^2t\).
\end{enumerate}
These conditions can be ascertained in polynomial time, and every valid state has encoding length polynomial in the game encoding and \(\log(1/\epsilon)\).

For the initial successor, we set \(\boldsymbol a^{\prime (1)}=\boldsymbol a'(\boldsymbol0)\), $r_{i,a_i}^{(1)}=\payoffdenom [u_i(a_i,\boldsymbol a_{-i}^{\prime (1)})-u_i(\boldsymbol a^{\prime (1)}) ]$ for every $i$ and $a_i \in \mathcal{A}_i$, and $\mu^{(1)}$ is the distribution supported solely on $\boldsymbol a^{\prime (1)}$. Then, $\operatorname{Next}(s^{(0)})
 = (1,\boldsymbol r^{(1)}, \mu^{(1)})$. For a valid state with \(1\le t<T\), we define \(\operatorname{Next}(s)\) by one regret matching step followed by the deterministic action selection of~\Cref{lem:zero-drift-selection} and the compression of~\Cref{lem:rational-compression}; when $t = T$, we set \(\operatorname{Next}(s)=s\) when \(t=T\). By~\cref{lem:zero-safe-step,lem:rational-compression}, this transition step takes polynomial time, and maps valid states to valid states. Moreover, the underlying potential for the $\PLS$ argument is taken as \(\operatorname{Val}(s)=t+1\), with \(\operatorname{Val}(s^{(0)})=1\). This trivially guarantees that every nonterminal state strictly improves the potential. We are now ready to prove~\Cref{thm:ivfcce-binary}.

\begin{proof}[Proof of~\Cref{thm:ivfcce-binary}]
We have defined a state space wherein each state has polynomial encoding length. We also provided an initial state, a polynomial-time successor algorithm, and an integer potential that strictly improves along successive states, unless the counter is equal to $T$.

For any valid \((t,\boldsymbol r,\mu)\), consistency gives
\(\barreg_i(\mu)=\max_{a_i\in\mathcal A_i} r_{i,a_i}/(\payoffdenom t)\). The invariance we have established in~\Cref{prop:compressed-trajectory} implies
\[
 0\le\barreg_i(\mu)
 =\frac{\max_{a_i\in\mathcal A_i} r_{i,a_i}}{\payoffdenom t}
 \le\sqrt{d/t}.
\]
At a local optimum \(t=T\), so the right-hand side is at most \(\epsilon\) by our choice of $T$. Thus, every local optimum decodes to an \((\epsilon, 0)\)-IVFCCE. We have therefore defined a polynomial local-search instance in the sense of \citet{Johnson88:How}, establishing membership in $\PLS$. Finally, the theorem follows from the fact that \(\PLS\cap\PPAD=\CLS\)~\citep{Fearnley23:Complexity}.
\end{proof}

\subsubsection{Hardness at exponential precision}
\label{subsec:ivfcce-hardness}

Complementing the previous membership result, we now show that the problem is unlikely to admit a polynomial-time algorithm. In particular, we show that it is at least as hard as the \emph{\(P\)-matrix linear complementarity problem}, whose complexity is a notorious open problem.

\begin{theorem}
    \label{theorem:PLCP}
    There is a polynomial-time reduction from P-matrix LCP to succinct normal-form games with a polynomial number of players and two actions per player. Moreover, the constructed game has the polynomial expectation property.
\end{theorem}

In particular, an algorithm polynomial in the description of the game and \(\log(1/\epsilon)\) would imply a polynomial-time algorithm for simple stochastic games. Indeed, \citet{Gartner05:Simple} reduce simple stochastic games to rational \(P\)-matrix LCPs of polynomial size. We proceed to lay out the reduction, starting with some basic definitions.

\paragraph{From complementarity to IVFCCEs.}

In what follows, we let \(d\) denote the dimension of the problem. For \(\mathbf M\in\R^{d\times d}\) and \(\boldsymbol q\in\R^d\), the \emph{linear complementarity problem} \(\operatorname{LCP}(\mathbf M,\boldsymbol q)\) asks for vectors $\boldsymbol z$ and $\boldsymbol w\in\R^d$ such that
\begin{equation}
 \boldsymbol w=\mathbf M\boldsymbol z+\boldsymbol q, \quad \boldsymbol z,\boldsymbol w \ge0,
 \quad z_iw_i=0 \text{ for all } i \in [d].
 \label{eq:lcp}
\end{equation}
This latter condition is known as \emph{complementarity} and is what makes this class of problems intractable. $\mathbf M$ is said to be a \emph{\(P\)-matrix} if all of its principal minors are positive definite; equivalently, \(\operatorname{LCP}(\mathbf M,\boldsymbol q)\) has a unique solution for every \(\boldsymbol q\)~\citep{Cottle09:Linear}. We will take \(\mathbf M\) and \(\boldsymbol q\) to be integral, which can be assumed without loss of generality because of rescaling.

For a subset \(S\subseteq[d]\), we define the \emph{complementary basis matrix} \(\mathbf K(S) \in \mathbb{R}^{d \times d}\) whose column is \(-\mathbf M_{\cdot j}\) when \(j\in S\) and the unit column \(\boldsymbol e_j\) otherwise. Let \(\delta(S) \defeq \det\mathbf M_{S,S}\). By convention, \(\delta(\varnothing)=1\). Since \(\mathbf M\) is a \(P\)-matrix, \(\delta(S) >0 \) and \(\det\mathbf K(S) =(-1)^{|S|}\delta(S) \).

We will think of $S$ as a candidate basis where $w_i = 0$ for $i \in S$ and $x_i = 0$ for $i \notin S$. For \(\boldsymbol x(S)=\mathbf K(S)^{-1}\boldsymbol q\), we set \(z_i(S) = x_i(S)\) and \(w_i(S)=0\) for \(i \in S\), while \(z_i(S)=0\) and \(w_i(S)=x_i(S)\) for \(i \notin S\). Simple algebra shows that \(\boldsymbol w(S)=\mathbf M\boldsymbol z(S)+\boldsymbol q\). Indeed, $\mathbf{K}(S) \boldsymbol{x}(S) = - \sum_{i \in S} \mathbf{M}_{\cdot i} x_i(S) + \sum_{i \notin S} \boldsymbol{e}_i x_i(S)$. Moreover, $ \sum_{i \in S} \mathbf{M}_{\cdot i} x_i(S) = \mathbf{M} \boldsymbol{z}(S)$ and $\boldsymbol{w}(S) = \sum_{i \notin S} \boldsymbol{e}_i x_i(S)$.

Now, for \(i \notin S \), we define \(\mathbf D_i(S)\) to have column \(\boldsymbol q\) in position \(i\), column \(\mathbf M_{\cdot j}\) when \(j\in S \), and unit column \(\boldsymbol e_j\) otherwise.
Cramer's rule gives $x_i(S) = \det \mathbf{D}_i(S)/\delta(S)$. As a result, $w_i(S) = \det \mathbf{D}_i(S)/\delta(S)$.  Moreover, let us consider the set $S \cup \{i \}$. Since $i \in S \cup \{i \}$, we have $z_i(S \cup \{i \}) = x_i(S \cup \{i \})$. Using Cramer's rule again, we find
\[
    x_i(S \cup \{i\}) = \frac{ (-1)^{|S|} \det \mathbf{D}_i(S \cup \{i \}) }{ (-1)^{|S| + 1} \delta(S \cup \{i \})} = - \frac{ \det \mathbf{D}_i(S) }{ \delta(S \cup \{i\}) }.
\]
In particular, let us consider the scaled vectors \(\widetilde{\boldsymbol z}(S)=\delta(S)\boldsymbol z(S)\) and \(\widetilde{\boldsymbol w}(S)=\delta(S)\boldsymbol w(S)\). We have shown that, for $i \notin S$,
\begin{equation}
 \widetilde w_i(S) = \det\mathbf D_i(S)
 \text{ and }
 \widetilde z_i(S \cup\{i\})=-\det\mathbf D_i(S).
 \label{eq:edge-determinant}
\end{equation}
Moreover, multiplying~\eqref{eq:lcp} by $\delta(S)$, it also follows that
\begin{equation}
 \widetilde{\boldsymbol w}(S)=\mathbf M\widetilde{\boldsymbol z}(S)+\delta(S)\boldsymbol q.
 \label{eq:scaled-lcp}
\end{equation}
We are now ready to define the game. For the purpose of normalization, we set $U \defeq \max\{2,\max_{i,j}|\mathbf M_{i j }|,\max_i|q_i|\}$ and $H \defeq (dU)^d$, so that Hadamard's inequality gives \(|\det\mathbf D_i(T)|\le H\) for any $S \subseteq [d]$. In this context, we define player \(i\)'s payoffs by
\begin{equation}
 u_i(\mathsf{0},\boldsymbol a_{-i})=\frac12
 \text{ and }
 u_i(\mathsf{1},\boldsymbol a_{-i})
 =\frac12 - \frac{\det\mathbf D_i(S(\boldsymbol a_{-i}))}{2H},
 \label{eq:determinant-payoff}
\end{equation}
where \(S(\boldsymbol a_{-i})=\{j\ne i:a_j = \mathsf{1}\}\). In particular, a pure profile of the game induces a subset \(S\subseteq[d]\), where \(i \in S\) means that player \(i\) plays action \(\mathsf{1}\). Because of rescaling, these payoffs lie in \([0,1]\).

In what follows, if $\mu$ is a distribution over subsets, we define \(\bar{\boldsymbol z}=\E_{S \sim \mu}[\widetilde{\boldsymbol z}(S)]\), \(\bar{\boldsymbol w} = \E_{S \sim \mu}[\widetilde{\boldsymbol w}(S)]\), and \(\bar{\delta} = \E_{S \sim \mu}[\delta(S)]\), so that averaging \eqref{eq:scaled-lcp} gives \(\bar{\boldsymbol w} = \mathbf M \bar{\boldsymbol z} + \bar{\delta} \boldsymbol q\). The following lemma is the core of the reduction, linking the deviation benefit for each player's action to the initial linear complementarity problem.

\begin{lemma}[Deviation benefit]
\label{lem:determinant-regret}
We have
\begin{equation}
 \E_{\boldsymbol a\sim\mu}[u_i(\mathsf{0},\boldsymbol a_{-i})-u_i(\boldsymbol a)]
 =-\frac{\bar z_i}{2H}
 \text{ and }
 \E_{\boldsymbol a\sim\mu}[u_i(\mathsf{1},\boldsymbol a_{-i})-u_i(\boldsymbol a)]
 =-\frac{\bar w_i}{2H}.
 \label{eq:determinant-regret}
\end{equation}
As a result, \(\barreg_i(\mu)=-\min\{\bar z_i,\bar w_i\}/(2H)\).
\end{lemma}

\begin{proof}
We consider a fixed set $S \subseteq [d]$. Suppose first that \(i\in S\), meaning that player \(i\) is currently playing \(\mathsf{1}\). Deviating instead to \(\mathsf{0}\) changes the payoff by
 \[
 \frac{\det\mathbf D_i(S\setminus\{i\})}{2H}
 =-\frac{\widetilde z_i(S)}{2H},
\]
where we used~\eqref{eq:edge-determinant}. If \(i\notin S\), deviating to \(\mathsf{0}\) does not result in a payoff change and \(\widetilde z_i(S)=0\). Averaging over \(S\sim\mu\) proves the first claim.

The second claim is symmetric. When \(i\notin S\), deviating to \(\mathsf{1}\) changes the payoff by
\(-\det\mathbf D_i(S)/(2H)=-\widetilde w_i(S)/(2H)\). When \(i\in S\), both the payoff change and \(\widetilde w_i(S)\) are zero, so averaging over $S \sim \mu$ also proves the second claim. Finally, the claim about the external regret follows because
\[
    \barreg_i(\mu) = \max \left\{ \E_{\boldsymbol a\sim\mu}[u_i(\mathsf{0},\boldsymbol a_{-i})-u_i(\boldsymbol a)], \E_{\boldsymbol a\sim\mu}[u_i(\mathsf{1},\boldsymbol a_{-i})-u_i(\boldsymbol a)] \right\}. \qedhere
\]
\end{proof}

The lemma basically establishes the reduction: the CCE condition requires \(\bar z_i, \bar w_i \ge0\), whereas information-value-freeness requires \(\min\{\bar z_i, \bar w_i \}\le 0 \implies \bar z_i \bar w_i = 0 \); this establishes complementarity. As a result, if \(\mu\) is an exact \(\IVFCCE\), then \(\boldsymbol z=\bar{\boldsymbol z}/ \bar{\delta} \) and \(\boldsymbol w=\bar{\boldsymbol w}/\bar{\delta} \) constitute a solution to~\eqref{eq:lcp}, because \(\bar{\delta} > 0\) is an average of positive principal minors.

Moreover, to argue about the case where $\mu$ is not an exact IVFCCE, we claim that one can compute a rational \(\epsilon_*>0\), with polynomial encoding length, such that an exact solution of the original problem can still be recovered in polynomial time from any polynomial-support \(\epsilon_*\)-\(\IVFCCE\).

\begin{restatable}[Rounding]{lemma}{rounding}
\label{lem:lcp-rounding}
From any polynomial-support distribution \(\mu\) in the determinant game satisfying \(|\barreg_i(\mu)|\le\epsilon_*\) for all \(i\), one can recover the exact solution of \(\operatorname{LCP}(\mathbf M,\boldsymbol q)\) in polynomial time.
\end{restatable}

This rounding argument is standard, and is deferred to \Cref{app:lcp-rounding}. To complete the proof of~\Cref{theorem:PLCP}, we need to argue about the polynomial expectation property. This essentially follows because for independent random columns \(\boldsymbol v_1,\ldots,\boldsymbol v_d\), we have
\begin{equation}
 \E\det(\boldsymbol v_1,\ldots,\boldsymbol v_d)
 =\det(\E\boldsymbol v_1,\ldots,\E\boldsymbol v_d),
 \label{eq:expected-determinant}
\end{equation}
by virtue of multilinearity. In particular, under a (rational) mixed-strategy profile \(\boldsymbol x\), let \(x_j\) denote the probability that player \(j\) chooses \(\mathsf{1}\). \(\overline{\mathbf D}_i(\boldsymbol x)\) is defined to have column \(\boldsymbol q\) in position \(i\) and column \((1-x_j)\boldsymbol e_j+x_j\mathbf M_{\cdot j}\) in position \(j\ne i\). Independence and \eqref{eq:expected-determinant} give
\begin{equation*}
 \E_{\boldsymbol a\sim\boldsymbol x}[u_i(\boldsymbol a)]
 =\frac12-\frac{x_i}{2H}\det\overline{\mathbf D}_i(\boldsymbol x).
\end{equation*}
This determinant can be computed in polynomial time. In particular, this implies that there is polynomial-time algorithm for computing an exact CE~\citep{Papadimitriou08:Computing,Jiang15:Polynomial}, and hence a CCE. Finally, we highlight again that, by virtue of the reduction of~\citet{Gartner05:Simple}, \Cref{theorem:PLCP} implies that simple stochastic games reduce to IVFCCEs.

\section*{Acknowledgments}

We are grateful to Kshipra Bhawalkar, Christopher Liaw, Grigoris Velegkas, and Aranyak Mehta for several helpful discussions. Part of this work was done when the authors were interns at Google Research. I.A. thanks Gabriele Farina for pointing out (in a 2022 discussion) that regret matching has nonnegative external regret. The proofs of this paper were developed with assistance from GPT 5.6 Sol. The authors take full responsibility for all content.

\bibliography{references}

\appendix

\section{Proof of the FTRL guarantee}
\label{app:fixed-step-ftrl}

For completeness, we provide the proof that FTRL has nonnegative external regret. We let $\mathcal{R}_i$ denote the regularizer of FTRL. By shifting, we can assume that $\min_{ \boldsymbol{x}_i \in \Delta(\mathcal{A}_i)} \mathcal{R}_i(\boldsymbol{x}_i) = 0$.

\ftrlnonnegative*

\begin{proof}
Let \(\boldsymbol u_i^{(t)}=(u_i(a_i,\boldsymbol x_{-i}^{(t)}))_{a_i \in\mathcal A_i}\). On round \(t\), player \(i\) selects the mixed strategy
\[
 \boldsymbol x_i^{(t)}
 \in\argmaxop_{\boldsymbol x_i\in\Delta(\mathcal A_i)}
 \left\{
 \sum_{\tau<t}\langle \boldsymbol x_i,
 \boldsymbol u_i^{(\tau)}\rangle
 -\frac{1}{\eta_i} \mathcal{R}_i(\boldsymbol x_i) 
 \right\}.
\]
We define
\[
 V_i^{(t)}
 \defeq \max_{\boldsymbol x_i\in\Delta(\mathcal A_i)}
 \left\{\sum_{\tau=1}^t
 \langle\boldsymbol x_i,\boldsymbol u_i^{(\tau)}\rangle
 - \frac{1}{\eta_i} \mathcal R_i(\boldsymbol x_i)\right\}.
\]
By definition of \(\boldsymbol x_i^{(t)}\), we have
\[
 V_i^{(t)}\ge V_i^{(t-1)}
 +\langle\boldsymbol x_i^{(t)},\boldsymbol u_i^{(t)}\rangle.
\]
Summing over \(t\) and using that $V_i^{(0)} = 0$ (since $\min_{ \boldsymbol{x}_i \in \Delta(\mathcal{A}_i)} \mathcal{R}_i(\boldsymbol{x}_i) = 0$), we obtain
\[
 \sum_{t=1}^T\langle\boldsymbol x_i^{(t)},\boldsymbol u_i^{(t)}\rangle
 \le V_i^{(T)}
 \le\max_{a_i \in\mathcal A_i}\sum_{t=1}^T u_{i,a_i}^{(t)}.
\]
This concludes the proof.
\end{proof}

\section{Rounding exponentially precise solutions}
\label{app:lcp-rounding}

Finally, we formalize the rounding procedure. We recall the claim below.

\rounding*

\begin{proof}
We set
\begin{equation}
 \Lambda=d^2UH
 \text{ and }
 \epsilon_*=\frac1{4H^2(\Lambda+1)}.
 \label{eq:hardness-accuracy}
\end{equation}
Both numbers have polynomial encoding length. We first convert the IVFCCE guarantee into an approximate LCP solution. We compute \(\bar{\boldsymbol z},\bar{\boldsymbol w}, \bar{\delta} \) exactly. Every principal minor is a positive integer, so \(\bar{\delta} \ge1\). By \cref{lem:determinant-regret}, the upper bound \(\barreg_i\le\epsilon_*\) implies \(\bar z_i,\bar w_i\ge-2H\epsilon_*\), while the lower bound \(\barreg_i\ge-\epsilon_*\) implies \(\min\{\bar z_i,\bar w_i\}\le2H\epsilon_*\). We set
\(\widehat{\boldsymbol z}=\bar{\boldsymbol z}/ \bar{\delta} \),
\(\widehat{\boldsymbol w}=\bar{\boldsymbol w}/ \bar{\delta} \), and
\(\eta=2H\epsilon_*\). We know that
\(\widehat{\boldsymbol w} = \mathbf M \widehat{\boldsymbol z}+\boldsymbol q\). Moreover, every coordinate is at least \(-\eta\), and the smaller coordinate in each pair \((\widehat z_i,\widehat w_i)\) is at most \(\eta\).

The next step is to extract a complementary basis. To do so, we set \(i\in S\) exactly when \(\widehat w_i \le \widehat z_i\), in which case \(\widehat w_i\) is nonbasic; otherwise \(\widehat z_i\) is nonbasic. Each selected nonbasic coordinate lies between \(-\eta\) and \(\eta\). If \(\boldsymbol x\) collects the basic coordinates and \(\boldsymbol y\) the nonbasic ones, then the LCP equation \(\widehat{\boldsymbol w} = \mathbf M \widehat{\boldsymbol z}+\boldsymbol q\) can be written as
\(\mathbf K(S) \boldsymbol x+\mathbf L(S) \boldsymbol y=\boldsymbol q\), where \(\mathbf L(S) \) is a matrix whose entries have magnitude at most \(U\).

To bound the error, we now compare with the exact basic solution \(\boldsymbol x(S) = \mathbf K(S)^{-1}\boldsymbol q\), which can be obtained by setting the nonbasic vector to zero. The error can be bounded as
\begin{equation}
    \label{eq:roundingerror}
    \|\boldsymbol x - \boldsymbol x(S) \|_\infty = \| \mathbf K(S)^{-1}\mathbf L(S) \boldsymbol{y} \|_\infty
 \le d^2UH\eta=\Lambda\eta.
\end{equation}
where we used the fact that \(\|\boldsymbol y\|_\infty\le\eta\) and every entry of $\mathbf{K}(S)^{-1}$ is upper bounded by $H$ (by Cramer's rule).

It remains to establish nonnegativity of $\boldsymbol{x}(S)$. For the sake of contradiction, suppose that some coordinate of $\boldsymbol{x}(S)$ were negative. By Cramer's rule, it can be expressed as the ratio of integers with denominator at most $H$, so it would be at most $-1/H$. However, the corresponding approximate basic coordinate is at least \(-\eta\). Combining with the error bound in~\eqref{eq:roundingerror}, this would imply $1/H \leq (\Lambda+1)\eta$, which contradicts the fact that $(\Lambda+1)\eta = 1/(2H)$ (by our choice of parameters in~\eqref{eq:hardness-accuracy}). We conclude that setting all nonbasic coordinates to zero gives an exact nonnegative complementary solution, which is the unique solution to \(\operatorname{LCP}(\mathbf M,\boldsymbol q)\).
\end{proof}
\end{document}